\documentclass[11pt]{article} 
\usepackage{times}
\usepackage{amsmath,amsfonts}
\usepackage{epsfig,amssymb,amstext,xspace,theorem}
\usepackage{graphicx,color} 

\newcommand{\np}{{\em NP}\xspace}
\newcommand{\nphard}{\np-hard\xspace} 

\DeclareMathOperator*{\Exp}{E}
\newcommand{\E}[2][{}]{\ensuremath{\Exp_{#1}\bigl[#2\bigr]}}

\newtheorem{theorem}{Theorem}[section]
\newtheorem{lemma}[theorem]{Lemma}

\newtheorem{claim}[theorem]{Claim}

\newtheorem{defn}[theorem]{Definition}

{\theorembodyfont{\rmfamily} \newtheorem{remark}[theorem]{Remark}}

\def\blksquare{\rule{2mm}{2mm}}
\def\qedsymbol{\blksquare}
\newcommand{\bg}[1]{\medskip\noindent{\bf #1}}
\newcommand{\ed}{{\hfill\qedsymbol}\medskip}
\newenvironment{proof}{\bg{Proof : }}{\ed}
\newenvironment{proofof}[1]{\bg{Proof of #1 : }}{\ed}

\newcommand{\R}{\ensuremath{\mathbb R}}

\newcommand{\A}{\ensuremath{\mathcal{A}}}

\newcommand{\I}{\ensuremath{\mathcal I}}
\newcommand{\F}{\ensuremath{\mathcal F}}
\newcommand{\D}{\ensuremath{\mathcal D}}

\newcommand{\Nc}{\ensuremath{\mathcal N}}

\newcommand{\Sc}{\ensuremath{\mathcal S}}
\newcommand{\Pc}{\ensuremath{\mathcal P}}

\newcommand{\OPT}{\ensuremath{\mathit{OPT}}}
\newcommand{\cost}{\ensuremath{\mathit{cost}}}

\newcommand{\argmin}{\ensuremath{\mathrm{argmin}}}

\newcommand{\frall}{\ensuremath{\text{ for all }}}

\newcommand{\es}{\ensuremath{\emptyset}}

\newcommand{\assign}{\ensuremath{\leftarrow}}

\newcommand{\poly}{\operatorname{\mathsf{poly}}}

\newcommand{\e}{\ensuremath{\epsilon}}

\newcommand{\gm}{\ensuremath{\gamma}}

\newcommand{\sse}{\subseteq}

\def\br#1{{{(#1)}}}
\newcounter{one} \newcounter{two}
\newcommand{\on}{\ensuremath{\mathrm{\Roman{one}}}}

\newcommand{\bsc}{{\small \textsf{RASC}}\xspace}
\newcommand{\dsc}{{\small \textsf{DSC}}\xspace}

\newcommand{\bscp}{(BSC-P)\xspace}
\newcommand{\bufl}{{\small \textsf{RAUFL}}\xspace}

\newcommand{\dufl}{{\small \textsf{DUFL}}\xspace}
\newcommand{\sufl}{{\small \textsf{SUFL}}\xspace}

\newcommand{\robopt}{\ensuremath{\mathit{OPT_{Rob}}}}

\newcommand{\tx}{\ensuremath{\tilde x}}
\newcommand{\ty}{\ensuremath{\tilde y}}

\newcommand{\hd}{\ensuremath{\hat d}}
\newcommand{\htf}{\ensuremath{\widehat f}}
\newcommand{\hx}{\ensuremath{\hat x}}
\newcommand{\hy}{\ensuremath{\hat y}}
\newcommand{\hz}{\ensuremath{\hat z}}

\newcommand{\hh}{\ensuremath{\widehat h}}
\newcommand{\hp}{\ensuremath{\widehat p}}
\newcommand{\hq}{\ensuremath{\widehat q}}

\newcommand{\by}{\ensuremath{\bar y}}
\newcommand{\bz}{\ensuremath{\bar z}}
\newcommand{\brr}{\ensuremath{\bar r}}
\newcommand{\sx}{\ensuremath{x^*}}
\newcommand{\sy}{\ensuremath{y^*}}
\newcommand{\sz}{\ensuremath{z^*}}

\newcommand{\ld}{\ensuremath{\lambda}}
\newcommand{\Ld}{\ensuremath{\Lambda}}
\newcommand{\kp}{\ensuremath{\kappa}}
\newcommand{\al}{\ensuremath{\alpha}}
\newcommand{\tht}{\ensuremath{\theta}}
\newcommand{\dt}{\ensuremath{\delta}}
\newcommand{\Dt}{\ensuremath{\Delta}}
\newcommand{\sg}{\ensuremath{\sigma}}
\newcommand{\w}{\ensuremath{\omega}}
\newcommand{\bo}{\ensuremath{\boldsymbol{0}}}

\newcommand{\ve}{\ensuremath{\varepsilon}}
\newcommand{\vro}{\ensuremath{\varrho}}
\newcommand{\hro}{\ensuremath{\hat\rho}}
\newcommand{\hkp}{\ensuremath{\hat\kp}}

\newcommand{\findmin}{{\small \textsf{FindMin}}\xspace}
\newcommand{\chancealg}{{\small \textsf{RiskAlg}}\xspace}
\newcommand{\budgalg}{{\small \textsf{SA-Alg}}\xspace}
\newcommand{\lb}{\ensuremath{\mathsf{LB}}}
\newcommand{\ub}{\ensuremath{\mathsf{UB}}}
\newcommand{\Gm}{\ensuremath{\Gamma}}
\newcommand{\VaR}{VaR\xspace}

\title{Algorithms for Probabilistically-Constrained Models of Risk-Averse Stochastic
Optimization with Black-Box Distributions} 

\author{
         Chaitanya Swamy\thanks{{\tt cswamy@math.uwaterloo.ca}.
         Dept. of Combinatorics and Optimization, Univ. Waterloo, Waterloo, ON N2L 3G1.
         Supported in part by NSERC grant 32760-06.}
}

\date{}

\begin{document}

\maketitle

\begin{abstract}
We consider various stochastic models that incorporate the notion of risk-averseness  
into the standard 2-stage recourse model, 
and develop novel techniques for solving the algorithmic problems arising in these models.  
A key notable feature of our work that distinguishes it from work in some other related
models, such as the (standard) budget model and the (demand-) robust model,
is that we obtain results in the {\em black-box} setting, that is, 
where one is given {\em only} sampling access to the underlying distribution. 
Our first model, which we call the {\em risk-averse budget model}, incorporates the notion
of risk-averseness via a {\em probabilistic constraint} 
that restricts the probability (according to the underlying distribution) with which the 
second-stage cost may exceed a given budget $B$ to at most a given input threshold $\rho$.
We also a consider a closely-related model 
that we call the {\em risk-averse robust model}, where we seek to minimize the first-stage
cost and the {\em $(1-\rho)$-quantile} (according to the distribution) of the second-stage
cost. 

We obtain approximation algorithms for a variety of combinatorial optimization
problems including the set cover, vertex cover, multicut on trees,  
min cut, and facility location problems, in the risk-averse budget and robust models with  
{\em black-box distributions}. 
We first devise a {\em fully polynomial approximation scheme} for solving the
{\em LP-relaxations} of a wide-variety of risk-averse budgeted problems. Complementing
this, we give a rounding procedure that lets us use existing LP-based approximation
algorithms for the 2-stage stochastic and/or deterministic counterpart of the problem to
round the fractional solution. Thus, we obtain near-optimal solutions to risk-averse
problems that preserve the budget approximately and incur a small blow-up of the
probability threshold (both of which are unavoidable). 
To the best of our knowledge, these are the {\em first approximation results} for problems    
involving {\em probabilistic constraints and black-box distributions}.
Our results extend to the setting with non-uniform scenario-budgets,
and to a generalization of the risk-averse robust model, where the goal is to minimize the
sum of the first-stage cost and a weighted combination of the expectation and the
$(1-\rho)$-quantile of the second-stage cost. 

\end{abstract}

\section{Introduction}
Stochastic optimization models provide a means to model uncertainty in
the input data where the uncertainty is modeled by a probability
distribution over the possible realizations of the actual data, called {\em scenarios}.  
Starting with the work of Dantzig~\cite{Dantzig55} and Beale~\cite{Beale55} in the 1950s,
these models have found increasing application in a wide variety of areas; see,
e.g.,~\cite{BirgeL97,RuszczynskiS03} and the references therein.  
An important and widely-used model in stochastic programming is the {\em 2-stage recourse 
model}: 
first, given the underlying distribution over scenarios, one may take some 
{\em first-stage} actions to construct an anticipatory part of the solution, $x$, incurring an
associated cost $c(x)$.  
Then, a scenario $A$ is realized according to the distribution, and one may take
additional {\em second-stage recourse actions} $y_A$ incurring a certain
cost $f_A(x,y_A)$. 
The goal in the standard 2-stage model is to minimize the total expected cost,
$c(x)+\E[A]{f_A(x,y_A)}$. 
Many applications come under this setting. 
An oft-cited motivating example is the {\em 2-stage stochastic facility location
problem}. A company has to decide where to set up its facilities  
to serve client demands. The demand-pattern is not known precisely at the outset,
but one does have some statistical information 
about the demands. 
The first-stage decisions consist of deciding which facilities to open initially, 
given the distributional information about the demands; 
once the client demands are realized according to this distribution, we can extend the
solution by opening more facilities, incurring their recourse costs. 
The recourse costs are usually higher than the original ones (e.g., because opening a 
facility later involves deploying resources with a small lead time), could be
different for the different facilities, and could even depend on the realized scenario. 

A common criticism of the standard 2-stage model is that the expectation measure fails to
adequately measure the  {\em ``risk''} associated with the first-stage decisions: 
two solutions with the same expected cost are valued equally. 
But in realistic settings, one also considers the risk involved in the 
decision. For example, in the stochastic facility location problem, given two
solutions with the same expected cost, one which incurs a moderate second-stage cost in
all scenarios, and one where there is a non-negligible probability that a ``disaster  
scenario'' with a huge associated cost occurs, a company would naturally \mbox{prefer the
former solution.} 

\medskip
\noindent {\bf Our models and results.\ }
We consider various stochastic models that incorporate the notion of risk-averseness  
into the standard 2-stage model 
and develop novel techniques for solving the algorithmic problems arising in these models.  
A key notable feature of our work that distinguishes it from work in some other related
models~\cite{GuptaRS04,DhamdhereGRS05},  
is that we obtain results in the {\em black-box} setting, that is, 
where one is given {\em only} sampling access to the underlying distribution. 
To better motivate our models, 
we first give an overview of some related models considered in the
approximation-algorithms literature that also embody the idea of risk-protection,
and point out why these models are ill-suited to the design of \mbox{algorithms in the  
black-box setting.} 

One simple and natural way of providing some assurance against the risk due to 
scenario-uncertainty is to provide bounds on the second-stage cost incurred in each
scenario. Two closely related models in this vein are the {\em budget model}, considered
by Gupta, Ravi and Sinha~\cite{GuptaRS04}, and the {\em (demand-) robust model},
considered by Dhamdhere, Goyal, Ravi and Singh~\cite{DhamdhereGRS05}. 
In the budget model, one seeks to minimize the expected total cost subject to the
constraint that the second-stage cost $f_A(x,y_A)$ incurred in every scenario $A$ be at
most some input budget $B$. (In general, one could have a different budget $B_A$ for each  
scenario $A$, but for simplicity we focus on the uniform-budget model.)
Gupta et al. considered the budget model in the {\em polynomial scenario} setting, where
one is given explicitly a list of all scenarios (with non-zero probability) and their
probabilities, thereby 
restricting their attention to distributions with a polynomial-size support. 
In the robust model considered by Dhamdhere et al.~\cite{DhamdhereGRS05}, which is more in
the spirit of robust optimization, the goal is to minimize
$c(x)+\max_A f_A(x,y_A)$. 
It is easy to see how the two models are related:
if one ``guesses'' the maximum second-stage cost $B$ incurred by the optimum, then the
robust problem essentially reduces to the budget problem with budget $B$, except that the 
second-stage cost term in the objective function is replaced by $B$ (which is a constant).  
Notice that 
it is not clear how to even {\em specify} problems with exponentially many scenarios in
the robust model. 
Feige et al.~\cite{FeigeJMM07} expanded the model of~\cite{DhamdhereGRS05} 
by considering 
exponentially many scenarios, where the scenarios are 
implicitly specified by a cardinality constraint. 
However, considering scenario-collections that are determined only by a cardinality
constraint 
seems rather specialized and somewhat artificial, 
especially in the context of stochastic optimization; e.g., 
in facility location, it is rather stylized (and overly conservative) to assume that 
{\em every} set of $k$ clients (for some $k$) may show up in the second-stage. 
We will consider a more general way of specifying (exponentially many) scenarios in robust
problems, where the input specifies a black-box distribution and the collection
of scenarios is then given by the support of this distribution.
We shall call this model the {\em distribution-based} robust-model.

Both the budget model and the (distribution-based) robust model 
suffer from certain common drawbacks.
A serious {\em algorithmic limitation} of both these models (see Section~\ref{lbounds}) is
that 
for almost any (non-trivial) stochastic problem (such as {\em fractional} stochastic set
cover with at most 3 elements, 3 sets, 3 scenarios), 
one cannot obtain {\em any} approximation guarantees in the black-box setting using any 
{\em bounded number of samples} (even allowing for a bounded violation of the budget in
the budget model). 
Intuitively, the reason for this is that there could be scenarios that occur with
vanishingly small probability that one will almost never encounter in our samples,
but which essentially force one to take certain first-stage actions in order to satisfy
the budget constraints in the budget model, or to obtain a low-cost solution in the robust
model. 
Notice also that 
both the budget and robust models adopt the conservative view that one needs to bound the
second-stage cost in {\em every} scenario, regardless of {\em how likely} it is for the
scenario to occur. 
(By the same token, they also provide the greatest amount of risk-aversion.) In contrast,
many of the risk-models considered in the finance and stochastic-optimization literature, 
such as the mean-risk model~\cite{Markowitz52}, value-at-risk (\VaR)
constraints~\cite{Prekopa95,Jorion96,Pritsker97}, 
conditional \VaR~\cite{RockafellarU02},
do factor in the probabilities of different scenarios.

Our models for risk-averse stochastic optimization address the above issues, and
significantly refine and extend the budget and robust models. 
Our goal is to come up with a model that is sufficiently rich in modeling power to allow
for black-box distributions, {\em and} in which one can obtain strong algorithmic results.
Our models are motivated by the observation (see Appendix~\ref{budget}) that it {\em is}
possible to obtain approximation guarantees in the budget model with black-box
distributions, if one allows the second-stage cost to exceed the budget with some 
``small'' probability $\rho$ (according to the underlying distribution). 
We can turn this solution concept around
and incorporate it into the model to arrive at the following. 
We are now {\em given} a probability threshold $\rho\in[0,1]$.
In our new budget model, which we call the {\em risk-averse budget model}, given a budget
$B$, we seek $(x,\{y_A\})$ so as to minimize $c(x)+\E[A]{f_A(x,y_A)}$ subject to the 
{\em probabilistic constraint} $\Pr_A[f_A(x,y_A)>B]\leq\rho$. 
The corresponding {\em risk-averse (distribution-based) robust model} seeks to minimize
$c(x)+Q_\rho[f_A(x,y_A)]$, 
where $Q_\rho[f_A(x,y_A)]$ is the $(1-\rho)$-quantile of $\{f_A(x,y_A)\}_{A\in\A}$, which
is the smallest number $B$ such that $\Pr_A[f_A(x)>B]\leq\rho$.
Notice that the parameter $\rho$ allows us to control the risk-aversion level and tradeoff
risk-averseness against conservatism (in the spirit of~\cite{BertsimasS04,SoZY06}).
Taking $\rho=1$ in the risk-averse budget model gives the standard 2-stage recourse
model, whereas taking $\rho=0$ in the risk-averse budget- or robust-models
recovers the standard budget- and robust models respectively. 
In the sequel, we treat $\rho$ as a constant that is not part of the input. 

We obtain approximation algorithms for a variety of combinatorial optimization
problems (Section~\ref{apps}) including the set cover, vertex cover, multicut on trees,  
min cut, and facility location problems, in the risk-averse budget and robust models with  
{\em black-box distributions}. 
We obtain near-optimal solutions that preserve the budget approximately and
incur a small blow-up of the probability threshold. (One should expect to
violate the budget here; otherwise, by setting very high first-stage costs, one    
would be able to solve the decision version of an \nphard problem!)
To the best of our knowledge, these are the {\em first approximation results} for problems  
with {\em probabilistic constraints and black-box distributions}.
Our results extend to the setting with non-uniform scenario-budgets,
and to a generalization of the risk-averse robust model, where the goal is to minimize
$c(x)$ plus a weighted combination of $\E[A]{f_A(x,y_A)}$ and 
$Q_\rho[f_A(x,y_A)]$. 
In the sequel, we focus primarily on the risk-averse budget model since results obtained 
this model essentially translate to the risk-averse robust model (the budget-violation 
can be absorbed into the approximation ratio).  

Our results are built on two components. First, and this is the technically more difficult 
component 
and our main contribution, we devise a {\em fully polynomial
approximation scheme} for solving the LP-relaxations of a wide-variety of risk-averse
problems (Theorem~\ref{chancethm}). 
We show that in the black-box setting, for a wide variety of 2-stage problems, for any
$\e,\kp>0$, 
in time $\poly\bigl(\frac{\ld}{\e\kp\rho}\bigr)$, one can compute (with high probability)
a solution to the LP-relaxation of the risk-averse budgeted problem, of cost at most
$(1+\e)$ times the optimum where the probability that the second-stage cost exceeds the
budget $B$ is at most $\rho(1+\kp)$. Here $\ld$ is the maximum ratio between the
costs of the same action in stage II and stage I (e.g., opening a facility or choosing 
a set). We show in Section~\ref{lbounds} that the dependence on $\frac{1}{\kp\rho}$,
and hence, the violation of the probability-threshold, is unavoidable in the black-box
setting. We believe that this is a general tool of independent interest that will find
application in the design of approximation algorithms for other discrete 
risk-averse stochastic optimization problems, 
and that our techniques will find use in solving other probabilistic programs.  

The second component is 
a simple rounding procedure (Theorem~\ref{round}) 
that complements (and motivates) the above approximation scheme. 
As we mention below, our LP-relaxation is a relaxation of even
the {\em fractional risk-averse problem} (i.e., where one is allowed to take fractional
decisions). We give a general rounding procedure to convert a solution to our
LP-relaxation to a solution to the fractional risk-averse problem 
losing a certain factor in the solution cost, budget, and the probability of
budget-violation. This allows us to then use an  LP-based ``local'' approximation
algorithm for the corresponding 2-stage problem to obtain an integer solution, where a
local algorithm is one that approximately preserves the LP-cost of each scenario.
In particular, for various covering problems, one can use the local $2c$-approximation
algorithm in~\cite{ShmoysS06}, which is obtained using an LP-based
$c$-approximation algorithm for the deterministic problem.    

We need to overcome various obstacles 
to devise our approximation scheme. 
The first difficulty faced in solving a probabilistic program such as ours, is that  
the feasible region of {\em even the fractional problem}
is a non-convex set. 
Thus, even in the {\em polynomial-scenario} setting, it is not clear how to solve (even)
the fractional risk-averse problem.   
(In contrast, in the standard 2-stage recourse model, the fractional problem can be easily
formulated and solved as a linear program (LP) in the polynomial-scenario setting.)  
We formulate an LP-relaxation (which is also a relaxation of the fractional problem), 
where we introduce a variable $r_A$ for every scenario $A$ that is supposed to indicate 
whether the budget is exceeded in scenario $A$. 
Correspondingly, we have two sets of decision variables 
to denote the decisions taken in scenario $A$ in the two cases respectively where the
budget is exceeded 
and where it is not exceeded. 
The constraints that enforce this 
semantics will of course be problem-specific, but a common constraint that figures in all
these formulations  
is $\sum_A p_Ar_A\leq\rho$, which captures our probabilistic 
constraint. 
This constraint, which {\em couples} the different scenarios, 
creates significant challenges in solving the LP-relaxation. (Again, notice the contrast 
with the standard 2-stage recourse model.) 
We get around the difficulty posed by this coupling constraint by taking the 
{\em Lagrangian dual} with respect to this constraint, introducing a dual variable
$\Dt\geq 0$. 
The resulting maximization problem (over $\Dt$) has a 2-stage minimization LP 
embedded inside it. 
Although this 2-stage LP does not belong to the class of problems defined
in~\cite{ShmoysS06,SwamyS05,CharikarCP05}, we prove that for any fixed
$\Dt$, this 2-stage LP can be solved to ``near-optimality'' using the {\em sample average 
approximation} (SAA) method. 
The crucial insight here is to realize that for the purpose of obtaining a near-optimal
solution to the risk-averse LP, it suffices to obtain a rather {\em weak} guarantee for
the 2-stage LP, where we allow for an additive error proportional to $\Dt$. 
This guarantee is specifically tailored so that it is weak enough that one can prove such 
a guarantee by showing ``closeness-in-subgradients'' and the analysis in~\cite{SwamyS05},
and yet can be leveraged to obtain a near-optimal solution to (the relaxation of) our
risk-averse problem.   
Given this guarantee, we show that 
one can efficiently find a suitable value for $\Dt$ 
such that the solution obtained for this $\Dt$ (via the SAA method) 
satisfies the desired guarantees.

\medskip
\noindent {\bf Related work.\ }
Stochastic optimization is a field with a vast amount of literature; we direct the reader 
to~\cite{BirgeL97,Prekopa95,RuszczynskiS03} 
for more information on the subject. We survey the work that is most relevant
to our work. Stochastic optimization problems have only recently been studied 
from an approximation-algorithms perspective. 
A variety of approximation results have been obtained in the 2-stage recourse model, 
but more general models, such as risk-optimization or probabilistic-programming models 
have received little or no attention. 

The (standard) budget model was first considered by Gupta et al.~\cite{GuptaRS04}, who
designed approximation algorithms for stochastic network design problems
in this model. 
Dhamdhere et al.~\cite{DhamdhereGRS05} introduced the demand-robust model (which we call
the robust model), and obtained algorithms for the robust versions of various combinatorial
optimization problems;   
some of their guarantees were later improved 
by Golovin et al.~\cite{GolovinGR06}.
All these works focus on the polynomial-scenario setting. Feige, Jain, Mahdian, and
Mirrokni~\cite{FeigeJMM07} considered the robust model with exponentially many scenarios
that are specified implicitly via a cardinality constraint, and derived approximation
algorithms for various covering problems in this more general model.

There is a large body of work in the finance and stochastic-optimization literature,
dating back to Markowitz~\cite{Markowitz52}, that deals with risk-modeling and
optimization; see  
e.g.,~\cite{RockafellarU02,AcerbiT02,RuszczynskiS05} and the references therein.  
Our risk-averse models are related to some models in finance. In fact, 
the probabilistic constraint that we use is called a {\em value-at-risk} (\VaR) constraint 
in the finance literature, and its use in risk-optimization is quite popular in finance
models; it has even been written into some industry regulations~\cite{Jorion96,Pritsker97}. 

Problems involving probabilistic constraints are called {\em probabilistic}
or {\em chance-constrained programs}~\cite{CharnesC59,Prekopa73} in the
stochastic-optimization literature, and have been extensively studied (see, e.g., 
Pr\'{e}kopa~\cite{Prekopa03}). 
Recent work in this area~\cite{CalafioreC06,NemirovskiS05,ErdoganI07} has focused on
replacing the original probabilistic constraint by more tractable constraints 
so that any solution satisfying the new constraints also satisfies the original
probabilistic constraint with high probability. Notice that this type of ``relaxation'' is  
{\em opposite} to what one aims for in the design of approximation algorithms, where we
want that every solution to the original problem remains a solution to the relaxation (but
most likely, not vice versa). Although some approximation results in the opposite
direction are obtained in~\cite{CalafioreC06,NemirovskiS05,ErdoganI07}, they are obtained
for very structured constraints of the type $\Pr_\xi[G(x,\xi)\notin C]\leq\rho$, where $C$
is a convex set, $\xi$ is a continuous random variable whose distribution satisfies a
certain concentration-of-measure property, and $G(.)$ is a bi-affine or convex mapping;
also the bounds obtained involve a relatively large violation of the probability threshold
(compared to our $(1+\kp)$-factor).
To the best of our knowledge, there is no prior work in the stochastic-optimization or
finance literature on 
the design of {\em efficient algorithms} with {\em provable worst-case guarantees} for
discrete risk-optimization or probabilistic-programming problems. 
In the Computer Science literature, \cite{KleinbergRT00} and ~\cite{GoelI99} consider the
stochastic bin packing and knapsack problems with probabilistic constraints that limit the
overflow probability of a bin or the knapsack, and obtained novel approximation algorithms
for these problems. Their results are however obtained for {\em specialized distributions}
where the item sizes are {\em independent} random variables following Bernoulli,
exponential, or Poisson distributions specified in the input. In the context of stochastic
optimization, this constitutes a rather stylized setting that is far from the black-box
setting.    

The work closest in spirit to ours is that of So, Zhang, and Ye~\cite{SoZY06}. 
They consider the problem of minimizing the first-stage cost plus a
risk-measure called the {\em conditional \VaR} (C\VaR)~\cite{RockafellarU02}. 
Their model interpolates between the 2-stage recourse model and the (standard) robust
model (as opposed to the budget model in our case). 
They give an approximation scheme for solving the LP-relaxations of a broad class of
problems in the black-box setting, using which they obtain approximation algorithms for
certain discrete optimization problems. 
Our methods are however quite different from theirs. In their model, the fractional 
problem yields a {\em convex program} and moreover, they are able to use a nice
representation theorem in~\cite{RockafellarU02} for the C\VaR measure to convert their
problem into a 2-stage problem and then adapt the methods in~\cite{CharikarCP05}.
In our case, 
the non-convexity inherent in the probabilistic constraint creates various difficulties
(first the non-convexity, then the coupling constraint) and we consequently need to work
harder to obtain our result.  
We  remark that our techniques can be used to solve a {\em generalization of their
model}, where we have the same objective function 
but also include a probabilistic budget constraint as in our risk-averse budget model.

We now briefly survey the approximation results in recourse models.
The first such approximation result appears to be due to Dye, Stougie, and
Tomasgard~\cite{DyeST03}.   
The recent interest and flurry of algorithmic activity in this area
can be traced to the work of Ravi and Sinha~\cite{RaviS04} and
Immorlica, Karger, Minkoff and Mirrokni~\cite{ImmorlicaKMM04}, which gave approximation 
algorithms for the 2-stage variants of various discrete optimization problems 
in the polynomial scenario~\cite{RaviS04,ImmorlicaKMM04} and 
{\em independent-activation}~\cite{ImmorlicaKMM04} settings. 
Approximation algorithms for 2-stage problems with black-box distributions were first
obtained by Gupta, P\'{a}l, Ravi and Sinha~\cite{GuptaPRS04}, and subsequently by Shmoys
and Swamy~\cite{ShmoysS06} (see also preliminary version~\cite{ShmoysS04}). 
Various other approximation results for 2-stage problems have since been obtained; 
see, e.g., the survey~\cite{SwamyS06}. 
Multistage recourse problems in the black-box model were considered
by~\cite{GuptaPRS05,SwamyS05}; 
both obtain approximation algorithms with guarantees that deteriorate with the number of 
stages, either exponentially~\cite{GuptaPRS05} (except for multistage Steiner tree 
which was also considered in~\cite{HayrapetyanST05}), 
or linearly~\cite{SwamyS05}; improved guarantees for set cover and vertex cover have
been subsequently obtained~\cite{Srinivasan07}.

Our approximation scheme makes use of the SAA method, which is a simple and appealing
method for solving stochastic problems that is quite often used in practice. 
In the SAA method one samples a certain number of scenarios to estimate the scenario
probabilities by their frequency of occurrence, and then solves the 2-stage problem
determined by this approximate distribution. 
The effectiveness of this method depends on the sample size (ideally, polynomial) required
to guarantee that an optimal solution to the SAA-problem is a provably near-optimal
solution to the original problem. 
Kleywegt et al.~\cite{KleywegtSH01} (see also~\cite{Shapiro03}) prove a bound that depends
on the variance of a certain quantity that need not be  polynomially bounded. 
Subsequently, Swamy and Shmoys~\cite{SwamyS05}, and Charikar et al.~\cite{CharikarCP05}
obtained improved (polynomial) sample-bounds for a large class of structured 2-stage
problems. 
The proof in~\cite{SwamyS05}, which also applies to multistage programs, is based on
leveraging approximate subgradients, and  
our proof makes use of portions of their analysis.
The proof of Charikar et al.~\cite{CharikarCP05} is quite different; it applies to 
2-stage programs but proves the stronger theorem that even approximate solutions to the
SAA problem translate to approximate solutions to the original problem.

\section{Preliminaries} \label{prelim}

Let $\R_+$ denote $\R_{\geq 0}$. Let $\|u\|$ denote the $\ell_2$ norm of $u$. 
The {\em Lipschitz constant} of a function $f:\R^m\mapsto\R$ is the smallest $K$ such that
$|f(x)-f(y)|\leq K\|x-y\|$.  
We consider convex minimization problems $\min_{x\in\Pc} f(x)$, where
$\Pc\sse\R_+^m$ with $\Pc\sse B(\bo,R)=\{x:\|x\|\leq R\}$ for a suitable $R$, \mbox{and
$f$ is convex.}

\begin{defn} \label{apsgrad}
Let $f:\R^m\mapsto\R$ be a function. We say that $d\in\R^m$ is a {\em subgradient} of $f$
at the point $u$ if the inequality $f(v)-f(u)\geq d\cdot(v-u)$ holds for every $v\in\R^m$.
We say that $\hd$ is an {\em $(\w,\xi)$-subgradient} of $f$ at the point
$u\in\Pc$ if for every $v\in\Pc$, we have $f(v)-f(u)\geq\hd\cdot(v-u)-\w f(v)-\w f(u)-\xi$. 
\end{defn}

The above definition of an $(\w,\xi)$-subgradient is slightly weaker than the notion of
an $\w$-subgradient as defined in~\cite{ShmoysS06}, where one requires that
$f(v)-f(u)\geq d\cdot(v-u)-\w f(u)$. But this difference is superficial; one could also
implement the algorithm in~\cite{ShmoysS06} using the weaker notion of an
$(\w,\xi)$-subgradient. It is well known (see~\cite{BorweinL00}) that a convex function has
a subgradient at every point. 
One can infer from Definition~\ref{apsgrad} that,
letting $d_x$ denote a subgradient of $f$ at $x$, the Lipschitz constant of $f$ is at most
$\max_x\|d_x\|$. 

Let $K$ be a positive number, and $\tau,\vro$ be two parameters with $\tau<1$. 
Let $N=\log\bigl(\frac{2KR}{\tau}\bigr)$. 
Let $G'_\tau\sse\Pc$ be a discrete set such that for any $x\in\Pc$, 
there exists $x'\in G'_\tau$ with $\|x-x'\|\leq\frac{\tau}{KN}$. 
Define $G_\tau=G'_\tau\cup\bigl\{x+t(y-x), y+t(x-y): x,y\in G'_\tau,\ t=2^{-i},\ i=1,\ldots,N\bigr\}$.
We call $G_\tau$ and $G'_\tau$, an {\em $\frac{\tau}{KN}$-net} and an
{\em extended $\frac{\tau}{KN}$-net} respectively of $\Pc$.
As shown in~\cite{SwamyS05}, if $\Pc$ contains a ball of radius $V$ (where $V\leq 1$
without loss of generality), then one can construct $G'_\tau$ so that
$|G_\tau|=\poly\bigl(\log(\frac{KR}{V\tau})\bigr)$. 
As mentioned earlier, our algorithms make use of the sample average approximation (SAA)
method. 
The following result from Swamy and Shmoys~\cite{SwamyS05}, which we have
adapted to our setting, will be our main tool for analyzing the SAA method.

\begin{lemma}[\cite{SwamyS05}] \label{SAAlem}
Let $\htf$ and $f$ be two nonnegative convex functions with Lipschitz constant at
most $K$ such that at every point $x\in G_\tau$, there exists a vector $\hd_x\in\R^m$ 
that is a subgradient of $\htf(.)$ and an $\bigl(\frac{\vro}{8N},\xi\bigr)$-subgradient  
of $f(.)$ at $x$. Let $\hx=\argmin_{x\in\Pc}\htf(x)$. Then,
$f(\hx)\leq(1+\vro)\min_{x\in\Pc}f(x)+6\tau+2N\xi$.  
\end{lemma}

\begin{lemma}[Chernoff-Hoeffding bound~\cite{Hoeffding63}] \label{chernoff}
Let $X_1,\ldots,X_N$ be iid random variables with each \mbox{$X_i\in[0,1]$} and
$\mu=\E{X_i}$. Then, 
$\Pr[\bigl|\frac{1}{N}\sum_i X_i-\mu\bigr|>\e]\leq 2e^{-2\e^2N}$.
\end{lemma}

\section{The risk-averse budgeted set cover problem: an illustrative example}
\label{budgetsc} 
Our techniques can be used to efficiently solve the risk-averse versions of a variety
of 2-stage stochastic optimization problems, both in the risk-averse budget and 
robust models. 
In this section, we illustrate the main underlying ideas by focusing on the risk-averse
budgeted set cover problem. In the risk averse budgeted set cover problem (\bsc), we are
given a universe $U$ of $n$ elements and a collection $\Sc$ of $m$ subsets of $U$.
The set of elements to be covered is uncertain: we are given a probability
distribution $\{p_A\}_{A\in\A}$ of scenarios, where each scenario $A$ specifies a subset
of $U$ to be covered. The cost of picking a set $S\in\Sc$ in
the first-stage is $w_S^\on$, and is $w_S^A$ in scenario $A$. The goal is to determine
which sets to pick in stage I and which ones to pick in each scenario so as to minimize
the expected cost of picking sets, subject to 
$\Pr_A[\text{cost of scenario }A>B]\leq\rho$, where $\rho$ is a constant that is not part
of the input. Notice that the costs $w_S^A$ are only revealed when we sample scenario
$A$; thus, the ``input size'', denoted by $\I$ is $O(m+n+\sum_S\log w_S+\log B)$. 

For a given (fractional) point $x\in\R^m$ with $0\leq x_S\leq 1$ for all $S$, define
$f_A(x)$ to be the minimum value of $w^A\cdot y_A$ subject to 
$\sum_{S:e\in S}y_{A,S}\geq 1-\sum_{S:e\in S}x_S$ for $e\in A$, and 
$y_{A,S}\geq 0$ for all $S$. Let $\Pc=[0,1]^m$. 
As mentioned in the Introduction, the set of feasible solutions to even the 
{\em fractional risk-averse problem} (where one can buy sets fractionally) 
is not in general a convex set. 
We consider the following LP-relaxation of the problem, which is a relaxation
of even the fractional risk-averse problem (Claim~\ref{relax}). 
Throughout we use $A$ to index the scenarios in $\A$, and $S$ to index the sets in $\Sc$. 
\begin{alignat}{3}
\min & \quad & \sum_{S} w_S^\on x_S\ + \sum_{A,S}p_A & \bigl(w_S^A y_{A,S}+w_S^Az_{A,S}\bigr)
\tag{RASC-P} \label{rscp} \\ 
\text{s.t.} && \sum_A p_Ar_A & \leq \rho \label{pbudg} \\[-0.5ex] 
&& \sum_{S:e\in S} \bigl(x_S+y_{A,S}\bigr)+r_A & \geq 1 \qquad && 
\frall A,e\in A, \label{ycov} \\  
&& \sum_{S:e\in S} \bigl(x_S+y_{A,S}+z_{A,S}\bigr) & \geq 1 \qquad && 
\frall A,e\in A, \label{zcov} \\[-0.5ex] 
&& \sum_S w_S^Ay_{A,S} & \leq B && \frall A \label{sbudg} \\
&& x_S,y_{A,S},z_{A,S},r_A & \geq 0 && \frall A,S. \label{noneg}
\end{alignat}
Here $x$ denotes the first-stage decisions. The variable $r_A$ denotes whether one exceeds
the budget of $B$ for scenario $A$, and the variables $y_{A,S}$ and $z_{A,S}$ denote
respectively the sets picked in scenario $A$ in the situations where one does not exceed
the budget (so $r_A=0$) and where one does exceed the budget (so $r_A=1$). Consequently,
constraint \eqref{sbudg} ensures that the cost of the $y_A$ decisions does not exceed the
budget $B$, and \eqref{pbudg} ensures that the total probability mass of
scenarios where one does exceed the budget is at most $\rho$. 
Let $\OPT$ denote the optimum value of \eqref{rscp}.

A significant difficulty faced in solving \eqref{rscp} is that the scenarios are no longer
{\em separable} given a first-stage solution, since constraint \eqref{pbudg} couples
the different scenarios. 
As a consequence, in order to specify a solution to
\eqref{rscp} one needs to compute a first-stage solution and give an explicit procedure
that computes $(y_A,z_A,r_A)$ in any given scenario $A$.  
In our algorithms however, we can avoid this complication because, as we show below, given
only the first-stage component of a solution to \eqref{rscp}, one can round it to a
first-stage solution to the fractional risk-averse problem (and then to an integer
solution) losing a small factor in the solution cost and the
probability-threshold. But observe that if we have a first-stage solution $x$ to the
fractional risk-averse problem with probability-threshold $P$ such that there exist
second-stage solutions yielding a total expected cost of $C$, then one can also easily
compute second-stage solutions that yield no greater total cost (and where
$\Pr[$second-stage cost $>B]\leq P$), by simply solving the LP $f_A(x)$ in each scenario
$A$.  
This implies that our algorithm for solving \eqref{rscp} only needs to return a
first-stage solution to \eqref{rscp} that can be extended to a near-optimum solution
(without specifying an explicit procedure to compute the second-stage solutions).   

We show (Theorem~\ref{chancethm}) 
that for any $\e,\kp>0$, one can efficiently compute a first-stage
solution $x$ for which there exist solutions $(y_A,z_A,r_A)$ in every scenario $A$ 
satisfying \eqref{ycov}--\eqref{noneg} such that 
$w^\on\cdot x+\sum_A p_Aw^A\cdot(y_A+z_A)\leq (1+2\e)\OPT$, and
$\sum_A p_Ar_A\leq\rho(1+\kp)$.
Complementing this, 
we give a simple rounding procedure based on the rounding theorem in~\cite{ShmoysS06} to
convert a fractional solution to \eqref{rscp} to an integer solution using an LP-based
$c$-approximation algorithm for the {\em deterministic} set cover (\dsc) problem,   
that is, an algorithm that returns a set cover of cost at most $c$ times the optimum of
the standard LP-relaxation for \dsc. 
We prove this rounding theorem 
first, in order to better motivate our goal of solving \eqref{rscp}.   

\begin{claim} \label{relax}
$\OPT$ is a lower bound on the optimum of the fractional risk-averse problem.
\end{claim}

\begin{proof}
We show that any solution $\hx$ to the fractional risk-averse problem can be mapped to a
solution to \eqref{rscp} of no greater cost. Let $\hy_A$ be 
such that $f_A(\hx)=w^A\cdot\hy_A$, so $\Pr[f_A(\hx)>B]\leq\rho$. 
We set $x=\hx$. For scenario $A$, if $f_A(\hx)\leq B$, we set
$r_A=0,\ y_A=\hy_A,\ z_A=\bo$. Otherwise, we set $r_A=1,\ y_A=\bo,\ z_A=\hy_A$.
It is easy to see that this yields a feasible solution to \eqref{rscp} of cost
$w^\on\cdot\hx+\sum_A p_Af_A(\hx)$.
\end{proof}

\begin{theorem}[Rounding theorem] \label{round}
Let $\bigl(x,\{(y_A,z_A,r_A)\}\bigr)$ be a solution satisfying
\eqref{ycov}--\eqref{noneg} of objective value
$C=\bigl(w^A\cdot x+\sum_A p_Aw^A\cdot(y_A+z_A)\bigr)$, and let
$P=\sum_A p_Ar_A$. Given any $\ve>0$, \mbox{one can obtain}
\begin{list}{(\roman{enumi})}{\usecounter{enumi} \itemsep=0ex \leftmargin=1ex}
\item a solution $\hx$ such that 
$w^\on\cdot\hx+\sum_A p_Af_A(\hx)\leq\bigl(1+\frac{1}{\ve}\bigr)C$ and
$\Pr_A\bigl[f_A(\hx)>(1+\frac{1}{\ve})B\bigr]\leq(1+\ve)P$; 
\item an integer solution $(\tx,\{\ty_A\})$ of cost at most $2c\bigl(1+\frac{1}{\ve}\bigr)C$
such that $\Pr_A\bigl[w^A\cdot\ty_A>2cB(1+\frac{1}{\ve})\bigr]\leq (1+\ve)P$
using an LP-based $c$-approximation algorithm for the deterministic set cover problem. 
\end{list}
Moreover, one only needs to know the first-stage solution $x$ to obtain 
$\hx$ and $\tx$. 
\end{theorem}

\begin{proof}
Set $\hx=\bigl(1+\frac{1}{\ve}\bigr)x$. Consider any scenario $A$. 
Observe that $(y_A+z_A)$ yields a feasible solution to the second-stage problem for
scenario $A$. Also, if $r_A<\frac{1}{1+\ve}$, then $\bigl(1+\frac{1}{\ve}\bigr)y_A$ also 
yields a feasible solution. Thus, we have $f_A(\hx)\leq w^A\cdot(y_A+z_A)$ and if
$r_A<\frac{1}{1+\ve}$ then we also have $f_A(\hx)\leq\bigl(1+\frac{1}{\ve}\bigr)B$.
So $w^\on\cdot x+\sum_Ap_Af_A(\hx)\leq\bigl(1+\frac{1}{\ve}\bigr)C$
and $\Pr\bigl[f_A(\hx)>(1+\frac{1}{\ve})B\bigr]\leq\sum_{A:r_A\geq\frac{1}{1+\ve}}p_A
\leq(1+\ve)P$. 

We can now round $\hx$ to an integer solution $(\tx,\{\ty_A\})$ using the
Shmoys-Swamy~\cite{ShmoysS06} rounding procedure (which only needs $\hx$) losing a
factor of $2c$ in the first- and second-stage costs. 
This proves part (ii). 
\end{proof}

\subsection{\boldmath Solving the risk-averse problem \eqref{rscp}} 
\label{lpsolve}

We now describe and analyze the procedure used to solve \eqref{rscp}.
First, we get around the difficulty posed by the coupling constraint \eqref{pbudg} in
formulation \eqref{rscp} 
by using the technique of Lagrangian relaxation. We take the Lagrangian dual of
\eqref{pbudg} introducing a dual variable $\Dt$ to obtain the following formulation. 
\begin{equation}
\max_{\Dt\geq 0} \quad -\Dt\rho+\Bigl(\min_{x\in\Pc} \quad h(\Dt;x)\ =\ w^\on\cdot x+ 
\sum_A p_Ag_A(\Dt;x)\Bigr) \tag{LD1} \label{ldp}
\end{equation} 
\vspace{-1ex}
\begin{equation}
\text{where} \ \
g_A(\Dt;x)\ =\ \min \ \sum_S w_S^A\bigl(y_{A,S}+z_{A,S}\bigr) + \Dt r_A \ \
\text{s.t.} \ \ \eqref{ycov}\text{--}\eqref{sbudg}, \quad 
y_{A,S},z_{A,S},r_A\geq 0\ \ \frall S. \tag{P} \label{scenp}
\end{equation}
It is straightforward to show via duality theory that \eqref{rscp} and \eqref{ldp} have
the same optimal value, 
and moreover that if $(\sx,\{\sy_A\},\{\sz_A\},\{r^*_A\})$ is an 
optimal solution to \eqref{rscp} and $\Dt^*$ is the optimal value for the dual
variable corresponding to \eqref{pbudg} then
$\bigl(\Dt^*;\sx,\{(\sy_A,\sz_A,r^*_A)\}\bigr)$ is an optimal solution to
\eqref{ldp}. 
Recall that $\Pc=[0,1]^m$. 
Let $\OPT(\Dt)=\min_{x\in\Pc} h(\Dt;x)$. 
So $\OPT=\max_{\Dt\geq 0}\bigl(\OPT(\Dt)-\Dt\rho\bigr)$.  
Let $\ld=\max\bigl\{1,\max_{A,S}(w^A_S/w^\on_S)\bigr\}$, which we assume is known. 
The main result of this section is as follows. 
Throughout, when we say ``with high probability'', we mean that a failure probability of  
$\dt$ can be ensured using $\poly\bigl(\ln(\frac{1}{\dt})\bigr)$-dependence on the sample
size (or running time). 

\begin{figure}[t!]
{\small
{\bf\boldmath \chancealg($\e,\gm,\kp$)} \quad [$\e\leq\kp<1$; the quantities $p^\br i,\
\cost^\br i,\ (y_A,z_A,r_A)$ are used only in the analysis.] \\[-12pt]
\begin{list}{C\arabic{enumi}.}{\usecounter{enumi} \itemsep=0pt \leftmargin=5ex}
\item Fix $\ve=\e/6,\ \zeta=\gm/4,\ \eta=\rho\kp/16$.
Also, set $\sg=\e/6,\ \gm'=\gm/4,\ \beta=\kp/8$, and $\rho'=\rho(1+3\kp/4)$.
Consider the $\Dt$ values $\Dt_0,\Dt_1,\ldots,\Dt_k$, where $\Dt_0=\gm'$,
$\Dt_{i+1}=\Dt_i(1+\sg)$ and $k$ is the smallest value such that
$\Dt_0(1+\sg)^k\geq\ub$. 
Note that $k=O\bigl(\log(\frac{\ub}{\gm'})/\sg\bigr)$.

\item For each $\Dt_i$, let 
$\bigl(x^\br i,\{(y^\br i_A,z^\br i_A,r^\br i_A)\}\bigr)
\assign\text{\budgalg}(\Dt;\ve,\eta,\zeta)$  
(here $(y^\br i_A,z^\br i_A,r^\br i_A)$ is an optimal solution to $g_A(\Dt_i;x^\br i)$ and
is implicitly given).  
Let $p^\br i=\sum_A p_Ar^\br i_A$ and  
$\cost^\br i=h(\Dt_i;x^\br i)=w^\on\cdot x^\br i+
\sum_Ap_A\bigl(w^A\cdot(y^\br i_A+z^\br i_A)+\Dt_ir^\br i_A\bigr)$. 

\item By sampling $n=\frac{1}{2\beta^2\rho^2}\ln\bigl(\frac{4k}{\dt}\bigr)$ scenarios, for 
each $i=0,\ldots,k$, compute an estimate $p'^{(i)}=\sum_A \hq_A r^\br i_A$ of $p^\br i$,
where $\hq_A$ is the frequency of scenario $A$ in the sampled set.

\item If $p'^{(0)}\leq\rho'$ then return $x^\br 0$ as the first-stage solution. [In
scenario $A$, return $(y_A,z_A,r_A)=(y^\br 0_A,z^\br 0_A,r^\br 0_A)$]. 

\item Otherwise (i.e., $p'^{(0)}>\rho'$) find an index $i$ such that $p'^{(i)}\geq\rho'$
and $p'^{(i+1)}\leq\rho'$ (we argue that such an $i$ must exist).
Let $a$ be such that $a\cdot p'{(i)}+(1-a)p'^{(i+1)}=\rho'$. Return the first-stage
solution $x=a\cdot x^\br i+(1-a)x^\br {i+1}$. [In scenario $A$, return the solution
$(y_A,z_A,r_A)=a(y^\br i_A,z^\br i_A,r^\br i_A)+
(1-a)(y^\br{i+1}_A,z^\br{i+1}_A,r^\br{i+1}_A)$.] 
\end{list}

\vspace{5pt}

{\bf\boldmath \budgalg($\Dt;\ve,\eta,\zeta$)} \quad [$K$ is (a bound on) the Lipschitz
constant of $h(\Dt;.)$; $\Pc\sse B(\bo,R)$ and $\Pc$ contains a ball of radius $V\leq 1$.] 
\\[-12pt] 
\begin{list}{B\arabic{enumi}.}{\usecounter{enumi} \itemsep=0pt \leftmargin=5ex}
\item Set $\tau=\zeta/6,\ N=\log\bigl(\frac{2KR}{V\tau}\bigr)$. Let $G_\tau\sse\Pc$ be an 
extended $\frac{\tau}{KN}$-net of $\Pc$ as defined in Section~\ref{prelim}, so that
$|G_\tau|=\poly\bigl(\log(\frac{KR}{V\tau})\bigr)$. 
Draw 
$\Nc=8N^2\bigl(\frac{4\ld}{\ve}+\frac{m}{\eta}\bigr)^2\ln\bigl(\frac{2|G_\tau|m}{\dt}\bigr)$ 
samples and for each scenario $A$, set $\hp_A=\Nc_A/\Nc$, where $\Nc_A$ is the number of
times scenario $A$ is sampled.

\item Solve the SAA problem $\min_{x\in\Pc} \hh(\Dt;x)$, where 
$\hh(\Dt;x)=w^\on\cdot x+\sum_A\hp_A g_A(\Dt;x)$ to obtain a solution $\hx$. Return $\hx$
and in scenario $A$, return the optimal solution to $g_A(\Dt;\hx)$.
\end{list}
}
\vspace{-12pt}
\caption{The procedures \chancealg and \budgalg.} 
\label{algdesc}
\end{figure}

\begin{theorem} \label{chancethm}
For any $\e,\gm,\kp>0$, \chancealg (see Fig.~\ref{algdesc}) runs in time 
$\poly\bigl(\I,\frac{\ld}{\e\kp\rho},\log(\frac{1}{\gm})\bigr)$, and returns with high
probability a first-stage solution $x$ and solutions $(y_A,z_A,r_A)$ for each scenario $A$
that satisfy \eqref{ycov}--\eqref{noneg} and such that 
(i) $w^\on\cdot x+\sum_A p_Aw^A\cdot(y_A+z_A)\leq(1+\e)\OPT+\gm$;
and (ii) $\sum_A p_Ar_A\leq\rho(1+\kp)$.
Under the very mild assumption ($*$) that $w^\on\cdot x+f_A(x)\geq 1$ for every
$A\neq\es,\ x\in\Pc$,%
\footnote{A similar assumption is made in~\cite{ShmoysS06} to obtain a multiplicative
guarantee.} 
we can convert this guarantee into a $(1+2\e)$-multiplicative
guarantee in the cost in time $\poly\bigl(\I,\frac{\ld}{\e\kp\rho}\bigr)$. 
\end{theorem}

Procedure \chancealg is described in Figure~\ref{algdesc}. 
In the procedure, we also specify the second-stage solutions for each scenario that can be
used to extend the computed first-stage solution to a near-optimal solution to
\eqref{rscp}. We use these solutions only in the analysis. 

We show in Section~\ref{lbounds} that 
the dependence on $\frac{1}{\kp\rho}$ 
is unavoidable in the black-box model.
The ``greedy algorithm'' for deterministic set cover~\cite{Chvatal79} is an LP-based  
$\ln n$-approximation algorithm, so Theorem~\ref{chancethm} combined with
Theorem~\ref{round} shows that for any $\e,\kp,\ve>0$ 
one can efficiently compute an integer solution $\bigl(\tx,\{\ty_A\}\bigr)$ of cost 
at most $2\ln n\bigl(1+\e+\frac{1}{\ve}\bigr)\cdot\OPT$ such that 
$\Pr_A\bigl[w^A\cdot\ty_A>2B\ln n(1+\e+\frac{1}{\ve})\bigr]\leq\rho(1+\kp+\ve)$.

Algorithm \chancealg is essentially a search procedure for the ``right''
value of the Lagrangian multiplier $\Dt$, wrapped around the SAA method, which is used in 
procedure \budgalg to compute a near-optimal solution to the minimization problem
$\min_{x\in\Pc} h(\Dt;x)$ for any given $\Dt\geq 0$. Theorem~\ref{budgthm} states the
precise approximation guarantee satisfied by the solution returned by \budgalg. 
Given this, we argue that 
by considering polynomially many $\Dt$ values that increase geometrically up to some 
upper bound $\ub$, one can find efficiently some $\Dt$ where the solution 
$\bigl(x,\{(y_A,z_A,r_A)\}\bigr)$ returned by \budgalg for $\Dt$ is such that 
$\sum_A p_Ar_A$ is ``close'' to $\rho$. This will also imply that this solution is a
near-optimal solution. 
We set $\ub=16(\sum_S w^\on_S)/\rho$, so $\log\ub$ is polynomially bounded.  
However, the search for the ``right'' value of $\Dt$ and our analysis are complicated by
the fact that we have two sources of error whose magnitudes we need to control: first, we
only have an approximate solution $\bigl(x,\{(y_A,z_A,r_A)\}\bigr)$ for $\Dt$, which also
means that one cannot use any optimality conditions; 
second, for any $\Dt$, we have only implicit access to the second-stage
solutions $\{y_A,z_A,r_A)\}$ computed by Theorem~\ref{budgthm}, so we cannot actually
compute or use $\sum_A p_Ar_A$ in our search, but will need to estimate it via sampling.  

\begin{theorem} \label{budgthm}
For any $\Dt\geq 0$, and any $\ve,\zeta,\eta>0$, \budgalg runs in time
$\poly\bigl(\I,\frac{\ld}{\ve\eta},\log(\frac{\Dt}{\zeta})\bigr)$ and returns, with high 
probability, a first-stage solution $x$ such that 
$h(\Dt;x)\leq (1+\ve)\OPT(\Dt)+\eta\Dt+\zeta$. 
\end{theorem}

\paragraph{Analysis.} 
For the rest of this section, $\e,\gm,\kp$ are fixed
values given by Theorem~\ref{chancethm}. We may assume without loss of generality 
that $\e\leq\kp<1$. 
We prove Theorem~\ref{budgthm} in Section~\ref{budgthmpf}. Here, we show how this leads to
the proof of Theorem~\ref{chancethm}. Given Theorem~\ref{budgthm} and
Lemma~\ref{chernoff}, we assume that the high probability event 
``$\forall i,\ \cost^\br i\leq(1+\ve)\OPT(\Dt_i)+\eta\Dt_i+\zeta
\text{ and }|p'^{(i)}-p^\br i|\leq\beta\rho$'' happens.

\begin{claim} \label{pend}
We have $p^\br k<\rho/2$ and $p'^{(k)}<\rho/2$.
\end{claim}

\begin{proof}
If $p^\br k>\frac{\rho(1+\ve)}{4}$, then 
$\cost^\br k-\eta\Dt_k>2(1+\ve)(\sum_S w^\on_S)>(1+\ve)\OPT(\Dt_k)+\zeta$,
which is a contradiction. 
The last inequality follows since $\OPT(\Dt)\leq\sum_S w^\on_S$ for any $\Dt$. 
Therefore, $p^\br k<\rho/2$, and 
$p'^{(k)}\leq p^\br k+\beta\rho<\rho/2$. 
\end{proof}

\begin{proofof}{Theorem~\ref{chancethm}}
Let $x$ be the first-stage solution returned by \chancealg, and $(y_A,z_A,r_A)$ be the
solution returned for scenario $A$. It is clear that
\eqref{ycov}--\eqref{noneg} are satisfied. 
Suppose first that $p'^{(0)}\leq\rho'$ (so $x=x^\br 0$.)
Part (ii) of the theorem follows since
$p^\br 0\leq p'^{(0)}+\beta\rho\leq\rho(1+\kp)$. 
Part (i) follows since 
$$
w^\on\cdot x^\br 0+\sum_Ap_Aw^A\cdot(y^\br 0_A+z^\br 0_A)
\leq h(\gm';x)\leq(1+\ve)\OPT(\gm')+\eta\gm'+\zeta
\leq(1+\ve)\OPT+\gm'(1+\ve+\eta)+\zeta.
$$ 
The penultimate inequality follows because for any $\Dt$, we have
$\OPT(\Dt)\leq\OPT(0)+\Dt\leq\OPT+\Dt$. 

Now suppose that $p'^{(0)}>\rho'$. In this case, there must exist an $i$ such that 
$p'^{(i)}\geq\rho'$, and $p'^{(i+1)}\leq\rho'$ because $p'^{(0)}>\rho'$ and 
$p'^{(k)}<\rho'$ (by Claim~\ref{pend}), so step C4 is well defined.
We again prove part (ii) first. We have
$\sum_A p_Ar_A=a\cdot p^\br i+(1-a)p^\br {i+1}\leq\rho'+\beta\rho\leq\rho(1+\kp)$. 
To prove part (i), observe that
$w^\on\cdot x+\sum_Ap_Aw^A\cdot(y_A+z_A)\leq
a\cdot\cost^\br i+(1-a)\cdot\cost^\br{i+1}-\Dt_i\bigl(a\cdot p^\br i+(1-a)\cdot p^\br{i+1}\bigr)$,
which is at most 
$$
(1+\ve)\Bigl(a\cdot\OPT(\Dt_i)+(1-a)\OPT(\Dt_{i+1})\Bigr)+\eta(a\Dt_i+(1-a)\Dt_{i+1})+\zeta
-\Dt_i(\rho'-\beta\rho).
$$
Now noting that $\Dt_{i+1}=(1+\sg)\Dt_i$, it is easy to see that
$\OPT(\Dt_{i+1})\leq(1+\sg)\OPT(\Dt_i)$. Also,
$\rho'-\beta\rho-\eta(1+\sg)\geq(1+\ve+2\sg)\rho$. So 
the above quantity is at most
$(1+\ve+2\sg)\bigl(\OPT(\Dt_i)-\Dt_i\rho\bigr)+\zeta\leq(1+\e)\OPT+\gm$. 

The running time is the time taken to obtain the solutions for all the
$\Dt_i$ values plus the time taken to compute $p'^{(i)}$ for each $i$.  
This is at most
$(k+1)\cdot\poly\bigl(\I,\frac{\ld}{\ve\eta},\log(\frac{\Dt_k}{\zeta})\bigr)+
O\bigl(\frac{\ln k}{\beta^2\rho^2}\bigr)$, using Theorem~\ref{budgthm}. 
Note that $\log(\Dt_k)$ is polynomially bounded. 
Plugging in $\ve,\eta,\zeta,\beta$, and $k$, we obtain the 
$\poly\bigl(\I,\frac{\ld}{\e\kp\rho},\log(\frac{1}{\gm})\bigr)$ bound.

\medskip \noindent
{\bf Proof of multiplicative guarantee.\ }
To obtain the multiplicative guarantee, we show that by initially sampling roughly  
$\max\{1/\rho,\ld\}$ times, with high probability, one can either determine that $x=\bo$  
is an optimal first-stage solution, or obtain a lower bound on $\OPT$ and then set $\gm$
appropriately in \chancealg to obtain the multiplicative bound.  
Recall that 
$f_A(x)$ is the minimum value of $w^A\cdot y_A$ over all $y_A\geq\bo$
such that $\sum_{S:e\in S}y_{A,S}\geq 1-\sum_{S:e\in S}x_S$ for $e\in A$.
Call $A=\es$ a null scenario.
Let $q=\sum_{A:A\neq\es}p_A$ and $\al=\min\{\rho,1/\ld\}$. Note that $\OPT\geq q$.
Let $\hz_A$ be an optimal solution to $f_A(\bo)$.
Define a solution $(\by_A,\bz_A,\brr_A)$ for scenario $A$ as follows. Set
$(\by_A,\bz_A,\brr_A)=(\bo,\bo,0)$ if $A=\es$, and $(\bo,\hz_A,1)$ if $A\neq\es$. 
We first argue that if $q\leq\al$, then $\bigl(\bo,\{(\by_A,\bz_A,\brr_A)\}\bigr)$ is an
optimal solution to \eqref{rscp}. It is clear that the solution is 
feasible since $\sum_A p_A\brr_A=q\leq\rho$. To prove optimality, suppose
$\bigl(\sx,\{(\sy_A,\sz_A,r^*_A)\}\bigr)$ is an optimal solution. Consider the solution
where $x=\bo$ and the solution for scenario $A$ is $(\bo,\bo,\bo)$ if $A=\es$, and
$(\bo,\sz_A+\sy_A+\sx,1)$ otherwise. This certainly gives a feasible solution. The
difference between the cost of this solution and that of the optimal solution is at most
$\sum_{A:A\neq\es}p_Aw^A\cdot\sx-w^\on\cdot\sx$, which is nonpositive since 
$w^A\leq\ld w^\on$ and $q\leq 1/\ld$. Setting $z_A=\hz_A$ for a non-null scenario can only
decrease the cost, and hence, also yields an optimal solution.

Let $\dt$ be the desired failure probability, which we may assume to be less than
$\frac{1}{2}$ without loss of generality.  
We determine with high probability if $q\geq\al$. We draw $M=\frac{\ln(1/\dt)}{\al}$
samples and compute $X=$ number of times a non-null scenario is sampled.
We claim that with high probability, if $X>0$ then
$\OPT\geq\lb=\frac{\dt}{\ln(1/\dt)}\cdot\al$; in this case, we return the solution
\chancealg{}$(\e,\e\lb,\kp)$ to obtain the desired guarantee. Otherwise, if $X=0$, we 
return $\bigl(\bo,\{(\by_A,\bz_A,\brr_A)\}\bigr)$ as the solution.

Let $r=\Pr[X=0]=(1-q)^M$. So $1-qM\leq r\leq e^{-qM}$. If
$q\geq\ln\bigl(\frac{1}{\dt}\bigr)/M$, then $\Pr[X=0]\leq\dt$, so with probability at
least $1-\dt$ we say that $\OPT\geq\lb$, which is true since $\OPT\geq q\geq\al$. If
$q\leq\dt/M$, then $\Pr[X=0]\geq 1-\dt$ and we return $\bigl(\bo,\{(\by_A,\bz_A,\brr_A)\}\bigr)$ 
as the solution, which is an optimal solution since $q\leq\al$. If
$\dt/M<q<\ln\bigl(\frac{1}{\dt}\bigr)/M$, then we always return a correct answer since it
is both true that $\OPT\geq q>\lb$, and that $\bigl(\bo,\{(\by_A,\bz_A,\brr_A)\}\bigr)$ is an
optimal solution. 
\end{proofof}

\subsubsection{Proof of Theorem~\ref{budgthm}} \label{budgthmpf}
Throughout this section, $\ve,\eta,\zeta$ are fixed at the values given in the
statement of Theorem~\ref{budgthm}. 
Let \bscp denote the problem $\min_{x\in\Pc}h(\Dt;x)$. The proof proceeds by analyzing the
subgradients of $h(\Dt;.)$ and $\hh(\Dt;.)$ and showing that Lemma~\ref{SAAlem} can be
applied here. 

We first note that the arguments given
in~\cite{ShmoysS06,SwamyS05,CharikarCP05} for 2-stage programs do not directly apply to
\bscp since it does not fall into the class of problems considered therein.
Shmoys and Swamy~\cite{ShmoysS06} show (essentially) that if one can compute an 
$(\w,\xi)$-subgradient of the objective function $h(\Dt;.)$ at any given point $x$ for a 
{\em sufficiently small} $\w,\xi$, then one can use the ellipsoid method to obtain a near 
optimal solution to \bscp.  
They argue that for a large class of 2-stage LPs, one
can efficiently compute an $(\w,\xi)$-subgradient using $\poly\bigl(\frac{\ld}{\w}\bigr)$
samples. 
Subsequently~\cite{SwamyS05}, they 
leveraged the proof of the ellipsoid-based algorithm to argue that the
SAA method 
also yields an efficient approximation scheme for 
the same class of 2-stage LPs. 
These proofs rely on the fact that for their class of 2-stage programs, each component of
the subgradient lies in a range bounded multiplicatively by a factor of $\ld$ and can
be approximated additively using $\poly(\ld)$ samples. 
However, in the case of \bscp, for a subgradient $d=(d_S)$ of $h(\Dt;.)$, we can only say
that $d_S\in[-w^A_S-\Dt,w^\on_S]$ (see Lemma~\ref{sub}), which makes it difficult 
to obtain an $(\w,\xi)$-subgradient using sampling for suitably small $\w,\xi$.
Charikar, Chekuri and P\'{a}l~\cite{CharikarCP05} considered a
similar class of 2-stage problems, and gave an alternate proof of efficiency of the SAA
method showing that even approximate solutions to the SAA problem translate to approximate
solutions to the original problem. Their proof shows that if $\Ld$ is such that
$g_A(\Dt;x)-g_A(\Dt;\bo)\leq\Ld w^\on\cdot x$ for every $A$ and $x\in\Pc$, then 
$\poly\bigl(\I,\frac{\Ld}{\e}\bigr)$ samples suffice to construct an SAA problem whose
optimal solutions correspond to $(1+\e)$-optimal solutions to the original problem. 
But for our problem, we can only obtain the bound $\Ld\leq{w^A\cdot x+\Dt(\sum_S x_S)}\leq
\ld w^\on\cdot x+\Dt\sum_S x_S$, and $\Dt$ might be large compared to $w^\on\cdot x$.

The key insight that allows us to circumvent these difficulties is that in order to 
establish our (weak) guarantee, where we allow for an additive error measured 
relative to $\Dt$, it suffices to be able to approximate each component $d_S$ of the
subgradient of $h(\Dt;.)$ within an additive error proportional to $(w^\on_S+\Dt)$, and
this can be done by drawing $\poly(\ld)$ samples. 
This enables one to argue that functions $\hh(\Dt;.)$ and $h(\Dt;.)$ satisfy the 
``closeness-in-subgradients'' property stated in Lemma~\ref{SAAlem}. 

The subgradients of $h(\Dt;.)$ and $\hh(\Dt;.)$ at $x$ are obtained from the
optimal dual solutions to $g_A(\Dt;x)$ for every $A$. The dual of $g_A(\Dt;x)$ is given by 
\begin{alignat}{3}
\max & \quad & \sum_e (\al_{A,e}+\beta_{A,e})\Bigl(1-\sum_{S:e\in S} x_S\Bigr) & 
-B\tht_A \tag{D} \label{scend} \\
\text{s.t.} && \sum_{e\in S\cap A}(\al_{A,e}+\beta_{A,e}) & \leq w^A_S(1+\tht_A) \qquad &&
\frall S \notag \\[-0.5ex] 
&& \sum_{e\in S\cap A}\beta_{A,e} & \leq w^A_S && \frall S \notag \\ 
&& \sum_{e\in A} \al_{A,e} & \leq \Dt \notag \\[-0.5ex] 
&& \al_{A,e},\beta_{A,e} & \geq 0 && \frall e\in A. \notag
\end{alignat}
Here $\al_{A,e}$ and $\beta_{A,e}$ are respectively the dual variables corresponding to
\eqref{ycov} and \eqref{zcov}, 
and $\tht_A$ is the dual variable corresponding to \eqref{sbudg}. 
As in~\cite{ShmoysS06}, we then have the following description of the subgradient of
$h$.   

\begin{lemma} \label{sub}
Let $(\al^*_A,\beta^*_A,\tht^*_A)$ be an optimal dual solution to $g_A(\Dt;x)$.  
Then the vector $d_x$ with components 
$d_{x,S}=w^\on_S-\sum_A p_A\sum_{e\in S}\bigl(\al^*_{A,e}+\beta^*_{A,e}\bigr)$ is a
subgradient of $h(\Dt;.)$ at $x$.
\end{lemma}

Since $\hh(\Dt;.)$ is of the same form as $h(\Dt;.)$, we have similarly that
$\hd_x=(\hd_{x,S})$, where 
$\hd_{x,S}=w^\on_S-\sum_A p_A\sum_{e\in S}\bigl(\al^*_{A,e}+\beta^*_{A,e}\bigr)$, is a 
subgradient of $\hh(\Dt;.)$ at $x$.  
Since $\hd_x$ and $d_x$ both have $\ell_2$ norm at most $\ld\|w^\on\|+|\Dt|$, $\hh(\Dt;.)$
and $h(\Dt;.)$ have Lipschitz constant at most $K=\ld\|w^\on\|+|\Dt|$. 

\begin{lemma} \label{apsub}
Let $d$ be a subgradient of $h(\Dt;.)$ at the point $x\in\Pc$, and
suppose that $\hd$ is a vector such that 
$\hd_S\in[d_S-\w w^\on_S-\xi/2m,d_S+\w w^\on_S+\xi/2m] \frall S$. Then $\hd$ is an
$(\w,\xi)$-subgradient of $h(\Dt;.)$ at $x$. 
\end{lemma}

\begin{proof}
Let $y$ be any point in $\Pc$. We have 
$h(\Dt;y)-h(\Dt;x)\geq\hd\cdot(y-x)+(d-\hd)\cdot(y-x)$.  
The second term is at least
$$
\sum_{S:d_S\leq\hd_S}(d_s-\hd_S)y_S+\sum_{S:\hd_S>d_S}(\hd_S-d_S)x_S
\geq \sum_S\bigl(-\w w^\on_S y_S-\w w^\on_S x_S\bigr)-\xi 
\geq -\w h(\Dt;y)-\w h(\Dt;x)-\xi. \\[-15pt]
$$
\end{proof}

In the sequel, we set $\w=\ve/8N,\ \xi=\eta\Dt/2N$.
Let $(\al^*_A,\beta^*_A,\tht^*_A)$ be the optimal dual solution to $g_A(\Dt;x)$ used to
define $\hd_x$ and $d_x$. Notice that $\hd_{x,S}$ is simply 
$w^\on_S-\sum_{e\in S}\bigl(\al^*_{A,e}+\beta^*_{A,e}\bigr)$ averaged over the scenarios
sampled independently to construct the SAA problem $\hh(\Dt;.)$, and $\E{\hd_{x,S}}=d_{x,S}$. 
The sample size $\Nc$ in \budgalg is specifically chosen so that the
Chernoff bound (Lemma~\ref{chernoff}) implies that 
$|\hd_{x,S}-d_{x,S}|\leq\w w^\on_S+\xi/2m\frall S$ with probability at least
$1-\frac{\dt}{|G_\tau|}$ for every $x\in G_\tau$; 
hence, $\hd_x$ is an $(\w,\xi)$-subgradient of $h(\Dt;.)$ at $x$ (by Lemma~\ref{apsub}). 
So taking the union bound shows that with probability at least $1-\dt$,
$\hh(\Dt;.)$ and $h(\Dt;.)$ satisfy the conditions of Lemma~\ref{SAAlem} with
$K=\ld\|w^\on\|+|\Dt|$, $\vro=\ve$ and $\xi$ (as above), which yields the desired
approximation guarantee. 

We can take $R=\sqrt{m}$ and $V=\frac{1}{2}$ here, 
so the number of samples $\Nc$ 
is $\poly\bigl(\I,\frac{\ld}{\ve\eta},\log(\frac{\Dt}{\zeta})\bigr)$. 
\hfill\qedsymbol

\begin{remark}
Notice that nowhere do we use the fact that the scenario-budgets are uniform, and thus,
our results (Theorem~\ref{budgthm} and hence, Theorem~\ref{chancethm}) extend to
the setting where we have different budgets for the different scenarios. The scenario
budgets $\{B^A\}$ are now not specified explicitly; we get to know $B^A$ when we sample
scenario $A$. (Notice that we may assume that $B^A\leq\ld\sum_S w^\on_S$ for all $A$.)
\end{remark}

\subsection{Risk-averse robust set cover} \label{robustsc}
In the risk-averse robust set cover problem, the goal is to choose some sets $x$ in stage
I and some sets $y_A$ in each scenario $A$ so that their union covers $A$,
so as to minimize $w^\on\cdot x+Q_\rho[w^A\cdot y_A]$. Recall that
$Q_\rho[w^A\cdot y_A]$ is the $(1-\rho)$-quantile of $\{w^A\cdot y_A\}_{A\in\A}$, that is,
the smallest $B$ such that $\Pr_A[w^A\cdot y_A>B]\leq\rho$. 
As mentioned in the Introduction, risk-averse robust problems can be essentially
reduced to risk-averse budget problems. We briefly sketch this reduction here for the set
cover problem. 
The same ideas can be used to obtain approximation algorithms for the risk-averse robust
versions of all the applications considered in Section~\ref{apps}.

We use the common method of ``guessing'' $B=Q_\rho[w^A\cdot y_A]$
for an optimal solution. Given this guess, we need to find integral
$\bigl(x,\{y_A\}\bigr)$ so as to minimize $w^\on\cdot x+B$ (and hence, $w^\on\cdot x$)
subject to the constraint that $x+y_A$ forms a set cover for $A$ and and 
$\Pr_A[w^A\cdot y_A>B]\leq\rho$. This looks very similar to the risk-averse budgeted set
cover problem; the only difference is that the expected second-stage cost does not appear
in the objective function. Thus, one can write an LP-relaxation for the (fractional)
risk-averse robust problem that looks similar to \eqref{rscp} except that the objective
function is now $w^\on\cdot x$, and constraint \eqref{zcov} and the variables $z_{A,S}$
can be dropped. After Lagrangifying \eqref{pbudg} using the dual variable $\Dt$, we obtain
the following problem 
\begin{equation}
\max_{\Dt\geq 0} \quad -\Dt\rho+\Bigl(\min \quad h'(\Dt;x)\ =\ w^\on\cdot x+
\sum_A p_Ag'_A(\Dt;x)\Bigr) \tag{LD2} \label{ldrobp}
\end{equation}
where 
$g'_A(\Dt;x)=\min\bigl\{\Dt r_A:\ \eqref{ycov},\eqref{sbudg},y_A\geq\bo,r_A\geq 0\bigr\}$.

Let $\robopt$ denote the optimum value of the {\em fractional} risk-averse robust problem 
$\min_{x\in\Pc}(w^\on\cdot x+Q_\rho[f_A(x)])$, and
$\robopt(B)$ denote the optimum value of \eqref{ldrobp} for a given $B\geq 0$.
Note that $\robopt(B)$ decreases with $B$. 
We prove that for any $B\geq 0$ and $\Dt\geq 0$, \budgalg returns a solution to the
inner minimization problem in \eqref{ldrobp} that satisfies the approximation guarantee
stated in Theorem~\ref{budgthm}. Arguing as in the proof of Theorem~\ref{chancethm}, this
implies that \chancealg can be used to obtain a near-optimal solution to \eqref{ldrobp}
while violating the probability threshold by a small factor. 

The claimed approximation guarantee for \budgalg follows because $h(\Dt;.)$ and its
sample-average approximation $\hh'(\Dt;.)$ constructed in \budgalg satisfy the
closeness-in-subgradients property of Lemma~\ref{SAAlem}.
Let $\al^*_{A,e}$ is the value of the dual variable corresponding to \eqref{ycov} in an
optimal dual solution to $g'_A(\Dt;x)$. Note that $\sum_e\al^*_{A,e}\leq\Dt$ for all $A$. 
Similar to Lemma~\ref{sub}, we now have that the vectors
$d_x=(d_{x,S})$ with $d_{x,S}=w^\on_S-\sum_Ap_A(\sum_{e\in S}\al^*_{A,e})$ and
$\hd_x=(\hd_{x,S})$ with $\hd_{x,S}=w^\on_S-\sum_A\hp_A(\sum_{e\in S}\al^*_e)$ are
respectively subgradients of $h'(\Dt;.)$ and $\hh'(\Dt;.)$ at $x$. 
Let $N,\Nc,\tau,G_\tau$ be as defined in \budgalg with $R=\sqrt{m}$, $V=\frac{1}{2}$ and
$K=\|w^\on\|+|\Dt|$. 
Using $\Nc$ samples, 
for any $x\in G_\tau$, with very high probability we have that 
$|\hd_{x,S}-d_{x,S}|\leq\eta\Dt/4mN$; thus, as in Lemma~\ref{apsub}, $\hd_x$ is an
$\bigl(0,\frac{\eta}{2N}\bigr)$-subgradient of $h'(\Dt;.)$ at $x$.  
So Lemma~\ref{SAAlem} shows that \budgalg returns a solution $\hx$
such that $h'(\Dt;x)\leq\OPT+\eta\Dt+\zeta$ with high probability. 
Notice that 
in fact, the approximation guarantee obtained via \budgalg is purely additive. 
Also, {\em one can avoid the dependence of the sample-size on $\ld$} (and $\ve$) here  
since the modified form of the subgradient means that we can ensure that
$|\hd_{x,S}-d_{x,S}|\leq\eta\Dt/4mn$ for every $x\in G_\tau$ and component $S$ using a
number of samples that is independent of $\ld$.
This implies that for any $\e,\gm,\kp>0$, \chancealg computes (nonnegative)
$\bigl(x,\{y_A,r_A\}\bigr)$  
satisfying \eqref{ycov}, \eqref{sbudg} such that 
$w^\on\cdot x\leq(1+\e)\robopt(B)+\gm$ and $\sum_A p_Ar_A\leq\rho(1+\kp)$.

To complete the reduction, we describe how to guess $B$. 
Let $W=\sum_S w^\on_S$, which is an upper bound on the optimum (with $\log W$ 
polynomially bounded). 
We use the standard method of enumerating values of $B$ increasing geometrically by
$(1+\e)$; we start at $\gm$ and end at the smallest value that is at least $W$. 
So if $B^*$ is the ``correct'' guess, then we are guaranteed to enumerate
$B'\in[B^*,(1+\e)B^*+\gm]$. We use \chancealg to compute the 
solution for each $B$, and return $\bigl(x,\{y_A,r_A\}\bigr)$ that minimizes 
$w^\on\cdot x+B$. Let $\bigl(x',\{y'_A,r'_A\}\bigr)$ be the solution computed for
$B'$. Then we have $w^\on\cdot x+B\leq w^\on\cdot x'+B'
\leq(1+\e)\robopt(B')+(1+\e)B^*+2\gm\leq(1+\e)\robopt+2\gm$. 
We remark that the same techniques yield a similar guarantee for the LP-relaxation of a  
generalization of the problem, where we wish to minimize $w^\on\cdot x$ plus a weighted
combination of $\E[A]{w^A\cdot y_A}$ and $Q_\rho[w^A\cdot y^A]$. 

We can convert the above guarantee into a purely multiplicative one under the same
assumption ($*$) stated in Theorem~\ref{chancethm}. 
Let $q=\sum_{A\neq\es}p_A$. Notice that if $q\leq\rho$, then $\robopt=0$ and $x=\bo$ is an 
optimal solution, and otherwise $\robopt\geq 1$. Let $\dt$ be such that
$(1+\kp)\frac{\dt}{\ln(1/\dt)}\leq 1$. Using $\frac{\ln(1/\dt)}{\rho'}$ samples we can
determine with high probability if $q\leq\rho'$ or if $q>\rho$. In the former case, we return
$x=\bo$ and $y_A$ in scenario $A$, where $y_A=\bo$ if $A=\es$ and is any feasible
solution if $A\neq\es$. Note that $w^\on\cdot x+Q_{\rho'}[w^A\cdot y_A]=0$. In the latter
case, we set $\gm=\e$, and obtain a execute the procedure detailed above to obtain a
$(1+3\e)$-multiplicative guarantee.

Finally, one can use Theorem~\ref{round} to round
the fractional solution 
to an integer solution, or to a solution to the fractional risk-averse robust problem.  
(The violation of the budget $B$ can now be absorbed into the approximation ratio.) 
For any $\e,\kp,\ve>0$, we obtain a fractional solution $\hx$ such that
$w^\on\cdot\hx+Q_{\rho(1+\kp+\ve)}[f_A(\hx)]\leq\bigl(1+\e+\frac{1}{\ve}\bigr)\robopt$,
and an integer solution $(\tx,\{\ty_A\})$ such that
$w^\on\cdot\tx+Q_{\rho(1+\kp+\ve)}[w^A\cdot\ty_A]\leq
2c\bigl(1+\e+\frac{1}{\ve}\bigr)\robopt$ using an LP-based $c$-approximation algorithm for
deterministic set cover. 

\medskip
Setting $B=0$ above yields a problem that is interesting in its own right. When $B=0$, we
seek a minimum-cost collection of sets $x$ that are picked {\em only in stage I} such that
$\Pr_A[x\text{ is not a set cover for }A]\leq\rho$. That is, we obtain a {\em
chance-constrained problem} without recourse. 
As shown above (although $B=0$ is not one of our ``guesses''), we can solve this
chance-constrained set cover problem to obtain a solution $x$ such that 
$w^\on\cdot x\leq(1+\e)\robopt(0)+\gm$ where  
$\Pr_A[x\text{ does not cover }A]\leq\rho(1+\kp)$.

\section{Applications to combinatorial optimization problems} \label{apps}
We now show that the techniques developed in Section~\ref{budgetsc} for the
risk-averse budgeted set cover problem can be used to obtain approximation algorithms for
the risk-averse versions of various combinatorial optimization problems such as covering
problems---(set cover,) vertex cover, multicut on trees, min $s$-$t$ cut---and facility
location problems.  
This includes many of the problems considered
in~\cite{GuptaPRS04,ShmoysS06,DhamdhereGRS05} in the standard 2-stage and demand-robust
models.  

In all the applications, the first step is to argue that procedure \chancealg can be used
to obtain a near-optimal solution to a suitable LP-relaxation of the problem while
violating the probability threshold by a small factor. Theorem~\ref{chancethm} proves this
for covering problems; for multicommodity flow and facility location, we need to modify
the arguments slightly.
The second step, which is more problem-specific, is to round the LP-solution to an integer
solution. Analogous to part (i) of Theorem~\ref{round}, we first round the LP-solution to
a solution to the fractional risk-averse problem. Given this, our task is now reduced to
rounding a fractional solution to a standard 2-stage problem into an integral one.
For this latter step, one can use any ``local'' LP-based approximation algorithm for the
2-stage problem, where a local algorithm is one that preserves approximately the cost of
each scenario. (For set cover, vertex cover and multicut on trees, we may use part (ii) of
Theorem~\ref{round} directly, which utilizes the local LP-rounding algorithm
in~\cite{ShmoysS06} (which in turn is obtained using an LP-based approximation algorithm
for the deterministic covering problem).) As in the case of risk-averse robust set cover,
our results extend to the setting of non-uniform budgets. 

We say that an algorithm is a $(c_1,c_2,c_3)$-approximation algorithm for the risk-averse
problem with budget $B$ and threshold $\rho$, if it returns a solution of cost at most
$c_1$ times the optimum where the probability that the second-stage cost exceeds 
$c_2\cdot B$ is at most $c_3\cdot\rho$. 

Our approximation results for the budgeted problem also translate to the risk-averse
robust version of the problem.  
Specifically, a $(c_1,c_2,c_3)$-approximation algorithm for the budgeted problem implies
that one can obtain an integer solution $(x,\{y_A\})$ to the robust problem such that
$c(x)+Q_{\rho(1+c_3)}[f_A(x,y_A)]\leq\max\{c_1,c_2\}\cdot\robopt$. 
As mentioned in Section~\ref{robustsc}, the robust problem with a guess of
$Q_\rho[f_A(x,y_A)]=0$ gives rise to a problem where one can take actions only in stage I
and one seeks to ``take care'' of ``most'' second-stage scenarios;
we can solve this chance-constrained problem approximately.
We also achieve bicriteria approximation guarantees for the problem of minimizing $c(x)$
plus a weighted combination of $\E[A]{f_A(x,y_A)}$ and $Q_\rho[f_A(x,y_A)]$.

\subsection{Covering problems} \label{covering}

\paragraph{Vertex cover and multicut on trees.} 
In the risk-averse budgeted vertex cover
problem, we are given a graph whose edges need to covered by vertices. The edge-set is
random and determined by a distribution (on sets of edges). A vertex $v$ may be picked  
in stage I or in a scenario $A$ incurring a cost of $w_v^\on$ or $w_v^A$ respectively. We
are also given a budget $B$ and a probability threshold $\rho$ and require that the
probability that the second-stage cost of picking vertices exceeds $B$ be at most $\rho$. 
In the risk-averse version of multicut on trees, we are given a tree, a (black-box)
distribution over sets of $s_i$-$t_i$ pairs, a budget $B$, and a threshold $\rho$. The
goal is to choose edges in stage I and in each scenario such that the union of edges
picked in stage I and in scenario $A$ forms a multicut for the $s_i$-$t_i$ pairs that
are revealed in scenario $A$. Moreover, the second-stage cost of picking edges may exceed
$B$ with probability at most $\rho$. The goal is to minimize the total expected cost. 

Both these problems are structured cases of risk-averse budgeted set cover. So one can
formulate an LP-relaxation of the risk-averse problem exactly as in \eqref{rscp} and by
Theorem~\ref{chancethm}, obtain a near-optimal solution to the relaxation.  
We may then apply Theorem~\ref{round} directly to these problems to round the
fractional solution. 
Since there is an LP-based 2-approximation algorithm for the deterministic versions of
both problems, we obtain the following theorem. 

\begin{theorem}
For any $\e,\kp,\ve>0$, there is a
$\bigl(4(1+\e+\frac{1}{\ve}),4(1+\e+\frac{1}{\ve}),1+\kp+\ve\bigr)$-approximation
algorithm for the risk-averse budgeted versions of vertex cover and multicut on trees.
\end{theorem}

\paragraph{\boldmath Min $s$-$t$ cut.} 
In the stochastic min $s$-$t$ cut problem, we are
given an undirected graph $G=(V,E)$ and a source $s\in V$. The location of the sink $t$ is
random and given by a distribution. We may pick an edge $e$ in stage I or in a scenario
$A$ incurring costs $w_e$ and $w_e^A$ respectively. The constraints are that in
any scenario $A$ with sink $t_A$, the edges picked in stage I and in that scenario induce
an $s$-$t_A$ cut, and the goal is to minimize the expected cost of choosing edges. In the
risk-averse budgeted problem there is the additional constraint that the
the second-stage cost may exceed a given budget $B$ with probability at most (a given
value) $\rho$. 

The LP-relaxation of the risk-averse problem based on a path-covering formulation is a
special case of \eqref{rscp}. The only additional observation needed to see that
Theorem~\ref{chancethm} can be applied here is that the covering problem
\eqref{scenp} for a scenario $A$ (and its dual) 
can be solved efficiently although there are an exponential number of constraints. 
Thus, procedures \chancealg and \budgalg can be implemented efficiently and 
we may obtain a near-optimal solution to the relaxation.

We use Theorem~\ref{round}, part (i) to convert the solution to a near-optimal solution
$\hx$ to the fractional risk-averse problem. We now use the algorithm in~\cite{DhamdhereGRS05},
which is a local LP-based $O(\log |V|)$-approximation algorithm 
to round this solution to an integral one. Their algorithm requires that there exist
multipliers $\ld_A$ in each 
scenario $A$ such that $w_e^A=\ld^Aw_e$ for every $e$; consequently we also need this for
our result. A detail worth noting is that their algorithm 
requires access also to the second-stage fractional solutions (but not the
scenario-probabilities). 
But this is not a problem since there are only polynomially many scenarios here
corresponding to the different locations of the sink. So given the first-stage solution 
$\hx$, one can simply compute the optimal fractional second-stage solution for each
scenario for use in their algorithm. 

\begin{theorem}
For any $\e,\kp,\ve>0$, there is an
$\bigl(O(\log|V|)(1+\e+\frac{1}{\ve}),O(\log|V|)(1+\e+\frac{1}{\ve}),1+\kp+\ve\bigr)$-approximation  
algorithm for risk-averse budgeted min $s$-$t$ cut. 
\end{theorem}

\subsection{Facility location problems} \label{facloc}
In the risk-averse budgeted facility location problem (\bufl), we have a set of $m$ facilities
$\F$, a client-set $\D$, and a distribution over client-demands. We may open facilities in
stage I or in a given scenario, and in each scenario $A$, for every client $j$ with
non-zero demand $d_j^A$, we must assign its demand to a facility opened in stage I or in
that scenario. The costs of opening a facility $i\in\F$ in stage I and in a scenario $A$
are $f_i^\on$ and $f_i^A$ respectively; the cost of assigning a client $j$'s demand in
scenario $A$ to a facility $i$ is $d_j^A c_{ij}$, where the $c_{ij}$'s form a metric. The
first-stage cost is the cost of opening facilities in stage I, and the cost of scenario
$A$ is the sum of all the facility-opening and client-assignment costs incurred in that
scenario. The goal is to minimize the total expected cost subject to the usual condition
that the probability that the second-stage cost exceeds $B$ is at most some threshold
$\rho$. For notational simplicity, we consider the case of $\{0,1\}$-demands, so a
scenario $A\sse\D$ simply specifies the clients that need to be assigned in that
scenario. We formulate the following LP-relaxation of the problem. Throughout, $i$ indexes
the facilities in $\F$ and $j$ the clients in $\D$.
\begin{alignat}{3}
\min & \quad & \sum_i f_i^\on y_i +  
\sum_{A\sse\D}p_A\Bigl(\sum_i f_i^A\bigl(y_{A,i}+v_{A,i}\bigr) & +
\sum_{j\in A,i}c_{ij}\bigl(x_{A,ij}+u_{A,ij}\bigr)\Bigr) \tag{RAFL-P} \label{ruflp} \\ 
\text{s.t.} && \sum_A p_Ar_A & \leq \rho \label{pbudg2} \\[-0.5ex] 
&& \sum_i x_{A,ij}+r_A & \geq 1 \qquad && \frall j\in A \label{asgnx} \\[-0.5ex]
&& \sum_i \bigl(x_{A,ij}+u_{A,ij}\bigr) & \geq 1 && \frall j\in A \label{asgnu} \\[-0.8ex]
&& x_{A,ij} & \leq y_i+y_{A,i} && \frall j\in A,i \label{fac1} \\
&& x_{A,ij}+u_{A,ij} & \leq y_i+y_{A,i}+v_{A,i} && \frall j\in A,i \label{fac2} \\
&& \sum_i f_i^Ay_{A,i}+\sum_{j\in A,i}c_{ij}x_{A,ij} & \leq B && \frall A \label{sbudg2}
\\[-0.5ex] 
&& y_i,y_{A,i},v_{A,i},x_{A,ij},u_{A,ij},r_A & \geq 0 && \frall A,i,j. \label{noneg2}
\end{alignat}
Here $y_i$ denotes the first-stage decisions. The variable $r_A$ denotes if one exceeds
the budget $B$ in scenario $A$; \eqref{pbudg2} limits the probability mass of such
scenarios to at most $\rho$. The decisions $(x_{A,ij},y_{A,i})$ and
$(u_{A,ij},v_{A,i})$ are intended to denote the decisions taken in scenario $A$ in the
two cases when does not exceed the budget, 
and when one does exceed the budget respectively. 
Correspondingly, \eqref{asgnx} and \eqref{asgnu} enforce that every client is assigned to 
a facility in these two cases, and \eqref{fac1} and \eqref{fac2} ensure that a client is
only assigned to a facility opened in stage I or in that scenario in these two cases.
Finally, \eqref{sbudg2} is the budget constraint for a scenario.

Let $\OPT$ be the optimal value of \eqref{ruflp}. Given first-stage decisions
$y\in[0,1]^m$, let $\ell_A(y)$ denote the minimum cost of fractionally opening facilities
and fractionally assigning clients in scenario $A$ to open facilities (i.e., facilities
opened to a combined extent of 1 in stage I and scenario $A$). Let $\Pc=[0,1]^m$. As in 
Section~\ref{budgetsc}, we Lagrangify \eqref{pbudg2} using a dual variable $\Dt\geq 0$ to
obtain the problem $\max_{\Dt\geq 0}\bigl(-\Dt\rho+\OPT(\Dt)\bigr)$ where 
$\OPT(\Dt)=\min_{y\in\Pc} h(\Dt;y)\bigr)$, 
$h(\Dt;y)=f^\on\cdot y+\sum_A p_Ag_A(\Dt;y)$, and $g_A(\Dt;y)$ is the minimum value of 
$\sum_i f^A_i(y_{A,i}+v_{A,i})+\sum_{j\in A,i}c_{ij}(x_{A,ij}+u_{A,ij})+\Dt r_A$ subject
to \eqref{asgnx}--\eqref{noneg2} (where the $y_i$'s are fixed now). 
As in Claim~\ref{relax}, it is easy to show that $\OPT$ is a lower bound on the optimal
value of even the fractional risk-averse problem. 

\begin{theorem} \label{ruflthm}
For any $\e,\gm,\kp>0$, in time 
$\poly\bigl(\I,\frac{\ld}{\e\kp\rho},\log(\frac{1}{\gm})\bigr)$, one can use \chancealg to  
compute (with high probability) $\bigl(y,\{(x_A,y_A,u_A,v_A,r_A)\}\bigr)$ 
that satisfies \eqref{asgnx}--\eqref{noneg2} with objective value 
$C\leq(1+\e)\OPT+\gm$ such that $\sum_A p_Ar_A\leq\rho(1+\kp)$.
This can be converted to a $(1+2\e)$-guarantee in the cost provided 
$f^\on\cdot y+\ell_A(y)\geq 1$ for every $y\in[0,1]^m,\ A\neq\es$. 
\end{theorem}

\begin{proof}
Examining procedure \chancealg, arguing that \chancealg can be used to approximately solve
\eqref{ruflp} involves two things:
(a) coming up with a bound $\ub$ such that $\log\ub$ is polynomially bounded so that one
can restrict the search for the right value of $\Dt$ in \chancealg; and  
(b) showing that an optimal solution to the SAA-version of the inner-minimization problem
for any $\Dt\geq 0$ constructed in \budgalg yields a 
solution to the true minimization problem 
that satisfies the approximation guarantee in Theorem~\ref{budgthm}. 

There are two notable aspects in which the risk-averse facility location 
differs from risk-averse set cover. First, unlike in set cover, one cannot 
ensure that the cost incurred in a scenario is always 0 by choosing the first-stage
decisions appropriately. Thus, the problem \eqref{ruflp} may in fact be infeasible. This
creates some complications in coming up with an upper bound \ub{} for use in \chancealg. We
show that one can detect by an initial sampling step that either the problem is
infeasible, or come up with a suitable value for \ub. Second, due to the non-covering
nature of the problem, one needs to delve deeper into the structure of the dual LP
for a scenario (after Lagrangifying \eqref{pbudg2}) to prove the closeness-in-subgradients
property for SAA objective function constructed in \budgalg and the true objective
function.

Define $C_A=\sum_{j\in A}(\min_i c_{ij})$. This is the minimum
possible assignment cost that one can incur in scenario $A$. 
We may determine with high probability using $O\bigl(\frac{1}{\rho\kp}\bigr)$ samples if 
$\Pr_A[C_A>B]>\rho$ or $\Pr_A[C_A>B]\leq\rho\bigl(1+\frac{5\kp}{28})$. In the former case,
we can conclude that the problem is infeasible. 
In the latter case, we set $\hro=\rho\bigl(1+\frac{5\kp}{28}\bigr)$, $\hkp$ such that
$\hro(1+\hkp)=\rho(1+\kp)$, and $\ub=\frac{32(1+\ve)(\sum_i f^\on_i+B)}{3\rho\kp}$, 
and call procedure \chancealg with these values of $\hro,\hkp$ and $\ub$ (and the given
$\e,\gm$). We prove in Claim~\ref{fl-pend} below that with this upper bound, 
$p^\br k,\ p'^{(k)}<\rho'=\hro(1+3\hkp/4)$; this is the only condition required for the
search for $\Dt$ in \chancealg. 

Task (b) boils down to showing that the objective function $\hh(\Dt;.)$ of the SAA-problem 
in \budgalg and the true problem $h(\Dt;.)$ satisfy the conditions of Lemma~\ref{SAAlem}. 
Due to the non-covering nature of the formulation, we need to derive additional insights
about optimal dual solutions to $g_A(\Dt;y)$ to prove this.  
Lemma~\ref{fl-close} proves that this holds with high probability, 
with $K=\ld\|f^\on\|+|\Dt|$, $\vro=\ve$ and $\xi=\frac{\eta\Dt}{2N}$.
So by Lemma~\ref{SAAlem}, the solution $\hy=\argmin_{y\in\Pc}\hh(\Dt;y)$ returned by
\budgalg satisfies the requirements of Theorem~\ref{budgthm}. 
As in the set cover problem, we may take $R=\sqrt{m}$, $V=\frac{1}{2}$, which ensures that
the sample size is polynomially bounded.
The proof of the conversion to a multiplicative guarantee is as in
Theorem~\ref{chancethm}.  
\end{proof}

Recall that $\Dt_k\geq\ub$ and $p^\br k=\sum_A p_Ar^\br k_A$, 
where $\bigl(y,\{(x_A,y_A,u_A,v_A,r_A)\}\bigr)$ is
the solution returned by \budgalg for $\Dt_k$ of cost
$h(\Dt_k;y)\leq(1+\ve)\OPT(\Dt_k)+\eta\Dt_k+\zeta$ with $\ve,\eta,\zeta$ set as in
\chancealg. 

\begin{claim} \label{fl-pend}
We have $p^\br k<\rho'$ and $p'^{(k)}<\rho'$, where $\rho'=\hro(1+3\hkp/4)$.
\end{claim}

\begin{proof}
Let $F=\sum_i f^\on_i$ and $q=\Pr_A[C_A>B]\leq\rho(1+5\kp/28)$. 
Consider the solution $y$ with $y_i=1$ for all $i$.  
For any $\Dt\geq 0$, we have 
$\OPT(\Dt)\leq h(\Dt;y)\leq F+\sum_A p_AC_A+q\Dt\leq F+B+q\Dt$.
Suppose $p^\br k\geq\rho'-\beta\hro$. Then
$\cost^\br k-\eta\Dt_k\geq\Dt_k\hro(1+9\hkp/16)\geq\Dt_k\rho(1+9\kp/16)$, where the last
inequality follows since $\hro(1+\hkp)=\rho(1+\kp)$ and $\hro\geq\rho$.
Also $(1+\ve)\OPT(\Dt_k)+\zeta\leq
2(1+\ve)(F+B)+(1+\ve)q\Dt_k<2(1+\ve)(F+B)+\Dt_k\rho(1+3\kp/8)$ since
$\ve=\e/6\leq\hkp/6\leq\kp/6$.  
But then $\cost^\br k-\eta\Dt_k>(1+\ve)\OPT(\Dt_k)+\zeta$ which gives a contradiction. 
So $p^\br k<\rho'-\beta\hro$, which implies that $p'^{(k)}<\rho'$.
\end{proof}

\begin{lemma} \label{fl-close}
With probability at least $1-\dt$, $\hh(\Dt;.)$ and $h(\Dt;.)$
satisfy the conditions of Lemma~\ref{SAAlem} with $K=\ld\|f^\on\|+|\Dt|$, $\vro=\ve$ and 
$\xi=\frac{\eta\Dt}{2N}$. 
\end{lemma}

\begin{proof}
Consider a point $y\in\Pc$. 
Consider an optimal dual solution to $g_A(\Dt;y)$ where
$\al^*_{A,j},\psi^*_{A,j},\beta^*_{A,ij},\Gm^*_{A,ij},\tht^*_A$ are the optimal values of 
the dual variables corresponding to \eqref{asgnx}--\eqref{sbudg2} respectively. 
Note that $g_A(\Dt;y)$ equals the objective value of this dual solution, which is given by
$$
\sum_{j\in A}\bigl(\al^*_{A,j}+\psi^*_{A,j}\bigr)-
\sum_i y_i\Bigl(\sum_{j\in A}\bigl(\beta^*_{A,ij}+\Gm^*_{A,ij}\bigr)\Bigr)-B\cdot\tht^*_A.
$$
We choose an optimal dual solution that minimizes $\sum_{i,j}\beta^*_{A,ij}$.
As in Lemma~\ref{sub}, it is easy to show that the vectors $\hd_y=(\hd_{y,i})$ and
$d_y=(d_{y,i})$ given by 
$\hd_{y,i}=f^\on_i-\sum_A\hp_A\sum_{j\in A}\bigl(\beta^*_{A,ij}+\Gm^*_{A,ij}\bigr)$ and
$d_{y,i}=f^\on_i-\sum_Ap_A\sum_{j\in A}\bigl(\beta^*_{A,ij}+\Gm^*_{A,ij}\bigr)$ are
respectively subgradients of $\hh(\Dt;.)$ and $h(\Dt;.)$ at $y$. 

Now we claim that for every $i$, $\sum_j \beta^*_{A,ij}\leq\Dt$ and 
$\sum_j\Gm^*_{A,ij}\leq f_i^A$. 
Given this, $\|\hd_y\|,\|d_y\|\leq K$ where $K=\ld\|f^\on\|+\Dt$ for any $y\in\Pc$, so $K$
is an upper bound on the Lipschitz constant of $h(\Dt;.)$ and $\hh(\Dt;.)$.  

The second inequality is a constraint of the dual, corresponding to variable
$v_{A,i}$. Suppose $\beta^*_{A,ij}>0$ for some $j$. The dual enforces the constraint 
$\al^*_{A,j}+\psi^*_{A,j}\leq c_{ij}(1+\tht^*_A)+\beta^*_{A,ij}+\Gm^*_{A,ij}$,
corresponding to variable $x_{A,ij}$.
We claim that this must hold at equality. By complementary slackness, we have
$x^*_{A,ij}=y_i+y^*_{A,i}$ where $(x^*_A,y^*_A,u^*_A,v^*_A)$ is an optimal primal solution
to $g_A(\Dt;y)$. So if $y_i>0$ then $x^*_{A,ij}>0$ and complementary slackness gives the
desired equality. If $y_i=0$ and the above inequality is strict, then we may decrease
$\beta^*_{A,ij}$ while maintaining dual feasibility and optimality, which gives a
contradiction to the choice of the dual solution. Thus, since the dual also imposes that 
$\psi^*_{A,j}\leq c_{ij}+\Gm^*_{A,ij}$ (corresponding to $u_{A,ij}$), we have that
$\beta^*_{A,ij}\leq\al^*_{A,j}$, so $\sum_j \beta^*_{A,ij}\leq\sum_j\al^*_{A,j}\leq\Dt$
(the last inequality follows from the dual constraint for $r_A$).

As in Lemma~\ref{apsub}, if $d$ is a subgradient of $h(\Dt;.)$ at $y$ and $\hd$
is a vector such that $|\hd_i-d_i|\leq \w f^\on_i+\frac{\xi}{2m}$, then $\hd$ is an
$(\w,\xi)$-subgradient of $h(\Dt;.)$ at $y$.

Since $\E{\hd_{y,i}}=d_{y,i}$ for every $y$ and $i$, plugging in the sample size $\Nc$
used in \budgalg and using the Chernoff bound (Lemma~\ref{chernoff}), we obtain 
with probability at least $1-\dt$, 
$|\hd_{y,i}-d_{y,i}|\leq\frac{\ve}{8N}f^\on_i+\frac{\eta\Dt}{4mN}\frall i$, 
for every point $y$ in the extended $\frac{\tau}{KN}$-net $G_\tau$ of $\Pc$.  
Thus, with probability at least $1-\dt$, $\hd_y$ is an
$\bigl(\frac{\ve}{8N},\frac{\eta\Dt}{2N}\bigr)$-subgradient of $h(\Dt;.)$ at $y$ for every
$y\in G_\tau$.
\end{proof}

We now discuss the rounding procedure. 
Analogous to part (i) of Theorem~\ref{round}, it is not hard to see that if
$\bigl(y,\{(x_A,y_A,u_A,v_A,r_A)\}\bigr)$ is a solution satisfying
\eqref{asgnx}--\eqref{noneg2} of objective value $C$ with $P=\sum_Ap_Ar_A$, 
then for any $\ve>0$, taking $\hy=y\bigl(1+\frac{1}{\ve}\bigr)$ gives
$\sum_i f_i\hy_i+\sum_A\ell_A(\hy)\leq\bigl(1+\frac{1}{\ve}\bigr)C$ and
$\Pr[\ell_A(\hy)>\bigl(1+\frac{1}{\ve}\bigr)B]\leq (1+\ve)P$. 
So now 
one can use a local approximation algorithm for {\em 2-stage stochastic facility location}
(\sufl) to round $\hy$.

Shmoys and Swamy~\cite{ShmoysS06} show that any LP-based $c$-approximation algorithm for
the deterministic facility location problem (\dufl) that satisfies a certain
``demand-obliviousness'' property can be used to obtain a
$\min\{2c,c+1.52\}$-approximation algorithm for \sufl, by using it in conjunction with the
1.52-approximation algorithm for \dufl in~\cite{MahdianYZ02}. 
``Demand-obliviousness'' means that the algorithm should round a fractional solution
without having any knowledge about the client-demands, and is imposed to handle the fact
that one does not have the second-stage solutions explicitly. 
There are some difficulties in applying this to our problem. First, the resulting
algorithm for \sufl need not be local. Secondly, and more significantly, even if we do
obtain a local approximation algorithm for \sufl by the conversion process
in~\cite{ShmoysS06}, the resulting algorithm may be {\em randomized}, if the
$c$-approximation algorithm for \dufl is randomized. This is indeed the case
in~\cite{ShmoysS06}; they obtain a randomized local 3.378-approximation algorithm using
the demand-oblivious, randomized 1.858-approximation algorithm of
Swamy~\cite{Swamy04}. (This was improved to a randomized local 3.25-approximation
algorithm by Srinivasan~\cite{Srinivasan07}, again using the algorithm in~\cite{Swamy04}.) 
Using such a randomized local $c'$-approximation algorithm for \sufl 
would yield a random integer solution such that there is at least a $1-\rho(1+\kp+\ve)$ 
probability mass in scenarios for which the {\em expected cost} incurred, where the 
expectation is over the random choices of the algorithm, is at most
$c'B\bigl(1+\frac{1}{\ve}\bigr)$.  
But we would like to make the stronger claim that, 
{\em with high probability over the random choices of the algorithm}, 
we return a solution where the probability-mass of scenarios with cost at most
$c'B\bigl(1+\frac{1}{\ve}\bigr)$ is at least $1-\rho(1+\kp+\ve)$.

We can take care of both issues by imposing the following (sufficient) condition on the
demand-oblivious algorithm for \dufl that is used to obtain an approximation algorithm for
\sufl (via the conversion process in~\cite{ShmoysS06}): we require that with probability
1, the algorithm return an integer solution where {\em each client's assignment 
cost} is within some factor of its cost in the fractional solution.   
One can use the randomized approximation algorithm of
Swamy~\cite{Swamy04} or the deterministic Shmoys-Tardos-Aardal (STA)
algorithm~\cite{ShmoysTA97}, both of which satisfy this condition. Given a 
fractional solution $(x,y)$ to \dufl with facility cost $F$, for a parameter
$\gm\in(0,1)$, the STA-algorithm returns an integer solution $(\tx,\ty)$ 
with facility cost is at most $F/\gm$,
where for every $j$, 
$\sum_i c_{ij}\tx_{ij}\leq\frac{3}{1-\gm}\cdot\sum_i c_{ij}x_{ij}$ (so for any demands,
the total assignment cost is at most $\frac{3}{1-\gm}$ times the fractional assignment cost).
Taking $\gm=\frac{1}{4}$ 
and applying the rounding procedure of~\cite{ShmoysS06} yields the following theorem. 

\begin{theorem}
For any $\e,\kp,\ve>0$, there is an
$\bigl(5.52(1+\e+\frac{1}{\ve}),5.52(1+\e+\frac{1}{\ve}),1+\kp+\ve\bigr)$-approximation
algorithm for risk-averse budgeted facility location.
\end{theorem}

\begin{remark}
The local approximation algorithm for \sufl developed by~\cite{RaviS04} is unsuitable for 
our purposes, since this algorithm needs to know explicitly the second-stage fractional
solution for each scenario, which is an exponential amount of information.
\end{remark}

\paragraph{Budget constraints on individual components of the second-stage cost.}  
Our techniques can be used to devise approximation algorithms for various fairly general 
risk-averse versions of facility location. Since the second-stage cost consists of two
distinct components, the facility-opening cost and the client-assignment cost, one can
consider risk-averse budgeted versions of the problem where we impose a {\em joint
probabilistic budget constraint} on the total second-stage cost, and each component of the
second-stage cost. 
That is, consider \eqref{ruflp} with the following additional constraints for each
scenario $A$: $\sum_i f^A_iy_{A,i}\leq B_F$ and $\sum_{j,i} c_{ij}x_{A,ij}\leq B_C$. Here
$B_F$ and $B_C$ are respectively budgets on the per-scenario facility-opening and
client-assignment costs. To put it in words, \eqref{ruflp} augmented with the above
constraints imposes the following joint probabilistic budget constraint:
\begin{gather*}
\begin{split}
\textstyle
\Pr_A\bigl[\text{total cost of scenario }A>B \ & 
\text{{\sc or} \ facility-cost of scenario $A>B_F$} \\ 
& \text{{\sc or} \ assignment-cost of scenario }A>B_C\bigr]\leq\rho.
\end{split}
\end{gather*}
Note that by setting the appropriate budget to $\infty$ we can model the absence of a 
particular budget constraint.
One can model various interesting situations by setting $B, B_F, B_C$ suitably. 
For example, suppose we set $B_F=0$ and $B=\infty$ (or equivalently $B=B_C$). Then we seek 
a minimum-cost solution where we want to choose facilities to open in stage I such that with probability at least
$1-\rho$, we can assign the clients in a scenario $A$ to (a subset of) the stage I
facilities while incurring assignment cost at most $B_C$.  
One can also consider risk-averse robust versions of the problem where we seek to minimize
the first-stage cost plus the $(1-\rho)$-quantile of a certain component of the
second-stage cost (i.e., the second-stage facility-opening, or assignment, or total
cost). 
Employing the usual ``guessing'' trick, this gives rise to a budget problem where we have
a budget constraint for a single component of the second-stage cost (that is, two of $B,
B_F$ and $B_C$ are set to $\infty$). As before, the guarantees obtained for the budget
problem (see below) translate to this risk-averse robust problem.

Our techniques can be used to solve this more general LP. Specifically,
Theorem~\ref{ruflthm} continues to hold. But here we face the
complication that even if we have a first-stage solution $x$ to the fractional risk-averse
problem for which we know that {\em there exist} second-stage feasible solutions that
yield a solution of total expected cost $C$, it is not clear how to compute such feasible
second-stage solutions. However, notice that \chancealg not only returns a first-stage
solution (with the above existence property) but also shows how to compute a suitable
second-stage solution in each scenario $A$, which thus, allows us to specify completely a 
near-optimal solution to the LP-relaxation (where the RHS of \eqref{pbudg2} is
$\rho(1+\kp)$). Whereas earlier we used these solutions only in the analysis, now they are
part of the algorithm. In the rounding procedure, the first step, where we convert the
solution to the LP-relaxation to a fractional solution to the risk-averse problem is
unchanged. But we of course now need a stronger notion of
``locality'' from our approximation algorithm for \sufl. We need an algorithm that
approximately preserves (with probability 1) both the facility-opening and
client-assignment components of the second-stage cost of each scenario. (Clearly, if the
budget constraint is imposed on only one of the components then we only need the
cost-preservation of that component.) 
Many LP-rounding algorithms for \sufl (such as the ones in~\cite{ShmoysS06,Srinivasan07})
do in fact come with this stronger local guarantee. Thus, one can use these to obtain an
approximation algorithm for the above risk-averse problem with multiple budget
constraints.  

Finally, we obtain the same approximation guarantees with non-uniform scenario
budgets $\{(B^A,B^A_F,B^A_C)\}$. 
The only small detail here is that in order to obtain the upper bound
\ub{} for use in \chancealg, we now determine if 
$\Pr[C_A>\min\{B^A_C,B^A\}]$ is greater than $\rho$ or at most
$\rho\bigl(1+\frac{5\kp}{28}\bigr)$. In the former case, we conclude infeasibility, and in
the latter, we set $\hro=\rho\bigl(1+\frac{5\kp}{28}\bigr)$, $\hkp$ such that
$\hro(1+\hkp)=\rho(1+\kp)$, and $\ub=\frac{32(1+\ve)(\sum_i f^\on_i+\sum_j\max_i
c_{ij})}{3\rho\kp}$ and run \chancealg with these values. (Note that we may assume that
$B^A_C\leq\sum_j\max_i c_{ij}$ for all $A$.)

\section{Sampling lower bounds} \label{lbounds}
We now prove various lower bounds on the sample size required to obtain a bounded
approximation guarantee for the risk-averse budgeted problem in the black-box model. 
We show that the dependence of the sample size on $\frac{1}{\kp}$ for an
{\em additive violation} of $\kp$ in the probability threshold {\em is unavoidable} in the
black-box model even for the fractional risk-averse problem and even if we allow a bounded
violation of the budget. 

The crux of our lower bounds is the following observation. 
Consider the following problem. We are given as input a threshold $\vro\in(0,\frac{1}{4})$
and a biased coin with probability $q$ of landing heads, where the coin is given as a
black-box; that is, we do not know $q$ but may toss the coin as many times as necessary to   
``learn'' $q$. The goal is to determine if $q\leq\vro$ or $q>2\vro$; if
$q\in(\vro,2\vro]$ then the algorithm may answer anything.
We prove that for any 
$\dt<\frac{1}{2}$, any algorithm that ensures error probability at most $\dt$ on every
input must need at least $\Nc(\dt;\vro)=\ln\bigl(\frac{1}{\dt}-1\bigr)/4\vro$ coin tosses
for each threshold $\vro$. 

\begin{lemma} \label{coin}
Let $\dt<\frac{1}{2}$ and $\A_{N(\dt;\vro)}$ be an algorithm that has failure probability
at most $\dt$ and uses at most $N(\dt;\vro)$ coin tosses for threshold $\vro$. Then,
$N(\dt;\vro)\geq\Nc(\dt;\vro)$ for every $\vro\in(0,\frac{1}{4})$.   
\end{lemma}

\begin{proof}
Suppose $N(\dt;\vro)<\Nc(\dt;\vro)$ for some $\vro\in(0,\frac{1}{4})$.
Let $X$ be a random variable that denotes the number of times the coin lands heads. If
$X=0$ then the algorithm must say ``$q\leq\vro$'' with probability at least $1-\dt$,
otherwise the algorithm errs with probability more than $\dt$ on $q=0$.
But then for some $q_0<\frac{1}{4}$ slightly greater than $2\vro$, we have
$\Pr[X=0]>(1-2\vro)^{\Nc(\dt;\vro)}\geq\frac{\dt}{1-\dt}$.  
So $\A$ will say ``$q\leq\vro$'' (and hence, err) for $q=q_0$, with probability more
than $\dt$. 
\end{proof}

As a corollary we obtain that for any $\dt<\frac{1}{2}$, it is impossible to determine if
$q=0$ or if $q>0$ with error probability at most $\dt$ using a bounded number of samples.

Now consider risk-averse budgeted set cover.
We say that a solution is an $(\e,\gm)$-optimal solution if its cost is at most
$(1+\e)\OPT+\gm$. Suppose there is an algorithm $\A$ for risk-averse budgeted set cover
that on any input (with a black-box distribution) draws a bounded number of samples and
returns an $(\e,\gm)$-optimal solution with probability at least 
$1-\dt,\ \dt<\frac{1}{2}$, where the probability-threshold is violated by at most $\kp$.  
Consider the following risk-averse budgeted set-cover instance.  
There are three elements $e_1,e_2,e_3$, three sets $S_i=\{e_i\},\ i=1,2,3$. 
The budget is $B\geq 6\gm$ and the probability threshold is $\rho\leq\frac{1}{8(1+\e)}$. 
The costs are $w^\on_{S_i}=B$ for all $i$, and $w^A_{S_1}=0,\ w^A_{S_2}=w^A_{S_3}=2B/3$
for every scenario $A$.
Let $\kp<\frac{1}{4}$.
There are 3 scenarios: $A_0=\es$, $A_1=\{e_1,e_2,e_3\}$, $A_2=\{e_2,e_3\}$ with 
$p_{A_1}=\rho-\kp,\ p_{A_3}=1-p_{A_1}-p_{A_2}$. 
Observe that if $p_{A_2}\leq\kp$, then $\OPT\leq\rho\cdot 4B/3$, and every
$(\e,\gm)$-optimal solution must have $x_{S_1}+x_{S_2}+x_{S_3}\leq\frac{1}{3}$. But if
$p_{A_2}>2\kp$ (which is possible since $\rho<1$) then any solution where the probability
of exceeding the budget is at most $\rho+\kp$ must have
$x_{S_2}+x_{S_3}\geq\frac{1}{2}$, otherwise the cost in both scenarios $A_1$ and $A_2$
will exceed $B$. Thus, algorithm $\A$ 
can be used to determine if $p_{A_2}\leq\kp$ or
$p_{A_2}>2\kp$. This is true even if we allow the budget to be violated by a factor 
$c<\frac{10}{9}$ since we must still have $x_{S_2}+x_{S_3}>\frac{1}{3}$ if
$p_{A_2}>2\kp$; 
choosing $B\gg 1,\ \rho\ll 1$, we can allow an arbitrarily large budget-violation.  
So since $\A$ has failure probability at most $\dt$, by
Lemma~\ref{coin}, it must draw $\Omega\bigl(\frac{1}{\kp}\bigr)$ samples.

Taking $\kp=0$ shows that obtaining guarantees without violating the probability threshold
is impossible with a bounded sample size, whereas taking $\kp=\kp\rho$ shows that a 
{\em multiplicative} $(1+\kp)$-factor violation of the probability threshold requires 
$\Omega\bigl(\frac{1}{\kp\rho}\bigr)$ samples. Moreover, taking $\rho=0$ shows that one
cannot hope to achieve any approximation guarantees in the (standard) budget model with
black-box distributions.

\begin{theorem} \label{lb}
For any $\e,\gm>0,\ \dt<\frac{1}{2}$, every algorithm for risk-averse budgeted set
cover that returns an $(\e,\gm)$-optimal solution with failure probability at most $\dt$
using a bounded number of samples
\begin{list}{$\bullet$}{\usecounter{enumi} \itemsep=-0.5ex \leftmargin=3ex}
\item must violate the probability threshold on some input;
\item requires $\Omega\bigl(\frac{1}{\kp}\bigr)$ samples if the probability-threshold is
violated by at most an additive $\kp$;
\item requires $\Omega\bigl(\frac{1}{\kp\rho}\bigr)$ samples if the probability-threshold
is violated by at most a multiplicative $(1+\kp)$-factor.
\end{list}
\end{theorem}

The proof of impossibility of approximation in the standard robust model with a bounded
sample size is even simpler. Consider the following set cover instance. We have a single
element $e$ that gets ``activated'' with some probability $p$; the cost of the set
$S=\{e\}$ is 1 in stage I and some large number $M$ in stage II. If $p=0$ then $\OPT=0$,
otherwise $\OPT=1$. Thus, it is easy to see that an algorithm returning an $(\e,\gm)$-optimal
solution can be used to distinguish between these two cases (it should set $x_S\leq\gm$ in
the former case, and $x_S$ sufficiently large in the latter).

\appendix

\section{A bicriteria approximation for the Shmoys-Swamy class of 2-stage stochastic LPs
in the standard budget model ($\rho=0$)} \label{budget}

Here we sketch how one can obtain a bicriteria approximation algorithm for the class of
2-stage LPs introduced in~\cite{ShmoysS06} in the standard budget model (that is, where we
have a {\em deterministic} budget constraint). We show that for any $\rho>0$, in time
inversely proportional to $\rho$, one can obtain a near-optimal solution where the total
probability-mass of scenarios where the budget is violated is at most $\rho$.
We consider the following class of 2-stage stochastic LPs~\cite{ShmoysS06}%
\footnote{This was stated in~\cite{ShmoysS04} with extra 
constraints $B^As_A\geq h^A$, but this is equivalent to
$\binom{B^A}{D^A}s_A+\binom{\bo}{T^A}r_A \geq \binom{h^A}{j^A}-\binom{\bo}{T^A}x$.}
in the standard budget model.
\begin{alignat}{3}
& \min \quad & h(x)\ =& \quad w^\on\cdot x\ & +\ \sum_Ap_A & f_A(x) \quad 
\text{subject to} \quad x\in\Pc\sse\R_+^m, \quad f_A(x)\leq B\ \ \frall A 
\tag{Stoc-P} \label{stclassp} \\ 
& \text{where}\ & f_A(x)\ =& \quad \min & w^A\cdot r_A &\ +\ q^A\cdot s_A \notag \\
& && \quad \ \ \text{s.t.} & D^As_A &\ +\ T^Ar_A \ \geq \ j^A-T^Ax \notag \\
& && & r_A, s_A &\ \geq\ \bo, \ \ r_A\in\R^m, s_A\in\R^\ell. \notag
\end{alignat}
Here (a) $T^A\geq\bo$ for every scenario $A$, and (b) for every $x\in\Pc$, 
$\sum_{A\in\A}p_Af_A(x)\geq 0$ and the primal and dual problems corresponding to
$f_A(x)$ are feasible for every scenario $A$. 
It is assumed that $\Pc\sse B(\bo,R)$, and that $\Pc$ contains a ball of radius $V$
($V\leq 1$) where $\ln\bigl(\frac{R}{V}\bigr)$ is polynomially bounded.  
Define $\ld=\max\bigl(1,\max_{A\in\A,S}\frac{w^A_S}{w^\on_S}\bigr)$; we assume that $\ld$
is known. Let $\OPT$ be the optimum value and $\I$ denote the input size. 

It is possible to adapt the proofs in~\cite{ShmoysS06,CharikarCP05,SwamyS06} to obtain
the bicriteria guarantee and one can also prove an SAA theorem in the style
of~\cite{SwamyS05,CharikarCP05}. But perhaps, the simplest proof, which we now describe,
is obtained using the ellipsoid-based algorithm in~\cite{ShmoysS06}.
Let $\Pc'=\{x\in\Pc: f_A(x)\leq B\frall A\}$. Note that unlike in the case where we have a
probabilistic budget constraint, $\Pc'$ is a {\em convex set}.

Consider running the ellipsoid-based algorithm in~\cite{ShmoysS06} with the following
modification. Suppose we wish to return a solution of value at most $(1+\e)\OPT+\gm$.
Let $N=\poly\bigl(m,\ln(\frac{KR}{V\gm})\bigr)$ be a suitably large value that is equal to
the number of iterations of the ellipsoid method. Let $\rho'=\rho/N$.
Suppose the center of the current ellipsoid is $x\in\Pc$. 
Using $O\bigl(\frac{1}{\rho'}\bigr)$ samples one can determine with high probability
if $\Pr_A[f_A(x)>B]>\rho'/2$ or if $\Pr_A[f_A(x)>B]\leq\rho'$. In the former case, by
sampling again $O\bigl(\frac{1}{\rho'}\bigr)$ times, with very high probability, we can
obtain a scenario $A$ such that $f_A(x)>B$. 
Now we compute a subgradient $d_{A,x}$ of
$f_A(.)$ (which is obtained from an optimal dual solution to $f_A(x)$) at $x$, and use the
inequality $d_{A,x}\cdot(y-x)\leq 0$ to cut the current ellipsoid. Notice that this is a
{\em valid inequality} since for any $y\in\Pc'$, by the definition of a subgradient, we
have $0<f_A(y)-f_A(x)\geq d_{A,x}\cdot(y-x)$. In the latter case, where we detect that
$\Pr_A[f_A(x)>B]\leq\rho'$, we continue as in the algorithm in~\cite{ShmoysS06}: we
mark the current point $x$ and use an approximate subgradient of $h(.)$ at $x$ to cut the
current ellipsoid. Proceeding this way we obtain a collection of marked points
$x_1,\ldots,x_k$, where $k\leq N$, such that with high probability,
$\Pr_A[f_A(x_i)>B]\leq\rho'$ for each $x_i$, and by the analysis in~\cite{ShmoysS06} we
have that $\min_i h(x_i)$ is ``close'' to $\OPT$.

The next step in the algorithm in~\cite{ShmoysS06} is to find a point in the convex hull
of $x_1,\ldots,x_k$ whose value is close to $\min_i h(x_i)$ (procedure \findmin). 
Notice that for any point $y$ in the convex hull of $x_1,\ldots,x_k$, we have
$\Pr_A[f_A(y)>B]\leq k\rho'\leq\rho$: for any scenario $A$ with $f_A(x_i)\leq B$ for all
$i$, the convexity of $f_A(.)$ implies that $f_A(y)\leq B$. 
Thus, although the set $\{x\in\Pc:\Pr_A[f_A(x)>B]\leq\rho\}$ is not convex, this does not
present a problem for us. So one can use procedure \findmin in~\cite{ShmoysS06} to return a
point $y$ such that $h(y)\leq(1+\e)\OPT+\gm$ where $\Pr_A[f_A(y)>B]\leq\rho$.


\begin{thebibliography}{10}

\bibitem{AcerbiT02}
C.~Acerbi and D.~Tasche. 
\newblock On the coherence of expected shortfall. 
\newblock {\em Journal of Banking and Finance}, 26,1487--1503, 2002.


\bibitem{Beale55}
E.~M.~L. Beale.
\newblock On minimizing a convex function subject to linear inequalities.
\newblock {\em Journal of the Royal Statistical Society, Series B}, 17:173--184;
  discussion 194--203, 1955.

\bibitem{BertsimasS04}
D.~Bertsimas and M.~Sim. 
\newblock The price of robustness, 
\newblock {\em INFORMS Journal on Operations Research}, 52:35--38, 2004. 

\bibitem{BirgeL97}
J.~R. Birge and F.~V. Louveaux.
\newblock {\em Introduction to Stochastic Programming}.
\newblock Springer-Verlag, NY, 1997.

\bibitem{BorweinL00}
J.~Borwein and A.~S. Lewis.
\newblock {\em Convex Analysis and Nonlinear Optimization}.
\newblock Springer-Verlag, NY, 2000.

\bibitem{CalafioreC06}
G.~Calafiore and M.~Campi.
\newblock The scenario approach to robust control design.
\newblock {\em IEEE Transactions on Automatic Control}, 51(5):742--753, 2006. 

\bibitem{CharikarCP05}
M.~Charikar, C.~Chekuri, and M.~P\'{a}l.
\newblock Sampling bounds for stochastic optimization.
\newblock {\em Proceedings, 9th RANDOM}, pages 257--269, 2005.

\bibitem{CharnesC59}
A.~Charnes and W.~Cooper.
\newblock Uncertain convex programs: randomized solutions and confidence levels. 
\newblock {\em Management Science}, 6:73--79,1959.

\bibitem{Chvatal79}
V. Chv{\'a}tal.
\newblock A greedy heuristic for the set-covering problem.
\newblock {\em Mathematics of Operations Research}, 4:233--235, 1979.

\bibitem{Dantzig55}
G.~B. Dantzig.
\newblock Linear programming under uncertainty.
\newblock {\em Management Science}, 1:197--206, 1955.

\bibitem{DhamdhereGRS05}
K.~Dhamdhere, V.~Goyal, R.~Ravi, and M.~Singh.
\newblock How to pay, come what may: approximation algorithms for demand-robust covering 
  problems. 
\newblock {\em Proceedings, 46th Annual {IEEE} Symposium on Foundations of Computer
  Science}, pages 367--378, 2005. 


\bibitem{DyeST03}
S.~Dye, L.~Stougie, and A.~Tomasgard.
\newblock The stochastic single resource service-provision problem.
\newblock {\em Naval Research Logistics}, 50(8):869--887, 2003.
\newblock Also appeared as ``The stochastic single node service provision problem'',
  COSOR-Memorandum 99-13, Dept. of Mathematics and Computer Science, Eindhoven 
  Technical University, Eindhoven, 1999.


\bibitem{ErdoganI07}
E.~Erdo\v{g}an and G.~Iyengar.
\newblock On two-stage convex chance constrained problems.
\newblock {Math. Methods of Operations Research}, 65(1):115--140, 2007.

\bibitem{FeigeJMM07}
U.~Feige, K.~Jain, M.~Mahdian, and V.~Mirrokni.
\newblock Robust combinatorial optimization with exponential scenarios. 
\newblock {\em Proceedings, 13th IPCO}, pages 439--453, 2007.


\bibitem{GoelI99}
A.~Goel and P.~Indyk.
\newblock Stochastic load balancing.
\newblock {\em Proceedings, 40th Annual IEEE Symposium on Foundations of Computer
  Science}, pages 579--586, 1999. 

\bibitem{GolovinGR06}
D.~Golovin, V.~Goyal, and R.~Ravi. 
\newblock Pay today for a rainy day: improved approximation algorithms for demand-robust
  min-cut and shortest path problems. 
\newblock {\em Proceedings, 23rd STACS}, pages 206--217, 2006.

\bibitem{GuptaPRS04}
A.~Gupta, M.~P\'{a}l, R.~Ravi, and A.~Sinha.
\newblock Boosted sampling: approximation algorithms for stochastic optimization.
\newblock {\em Proceedings, 36th Annual {ACM} Symposium on Theory of Computing},
  pages 417--426, 2004. 

\bibitem{GuptaPRS05}
A.~Gupta, M.~P\'{a}l, R.~Ravi, and A.~Sinha.
\newblock What about Wednesday? Approximation algorithms for multistage stochastic
  optimization.
\newblock {\em Proceedings, 8th APPROX}, pages 86--98, 2005.

\bibitem{GuptaRS04}
A.~Gupta, R.~Ravi, and A.~Sinha.
\newblock An edge in time saves nine: LP rounding approximation algorithms for stochastic
  network design. 
\newblock In {\em Proceedings, 45th Annual {IEEE} Symposium on Foundations of Computer
  Science}, pages 218-227, 2004. 

\bibitem{HayrapetyanST05}
A.~Hayrapetyan, C.~Swamy, and \'{E}.~Tardos.
\newblock Network design for information networks.
\newblock {\em Proceedings, 16th SODA}, pages 933--942, 2005.

\bibitem{Hoeffding63}
W.~Hoeffding.
\newblock Probability inequalities for sums of bounded random variables.
\newblock {\em Journal of the American Statistical Association}, 58:13--30, 1963.

\bibitem{ImmorlicaKMM04}
N.~Immorlica, D.~Karger, M.~Minkoff, and V.~Mirrokni.
\newblock On the costs and benefits of procrastination: approximation algorithms for
  stochastic combinatorial optimization problems.
\newblock {\em Proceedings, 15th Annual {ACM}-{SIAM} Symposium on Discrete
  Algorithms}, pages 684--693, 2004. 

\bibitem{Jorion96}
P.~Jorion. 
\newblock {\em Value at Risk: A New Benchmark for Measuring Derivatives Risk}. 
\newblock Irwin Professional Publishers, New York, 1996. 

\bibitem{KleinbergRT00}
J.~Kleinberg, Y.~Rabani, and {\'E}.~Tardos.
\newblock Allocating bandwidth for bursty connections.
\newblock {\em SIAM Journal on Computing}, 30(1):191--217, 2000.

\bibitem{KleywegtSH01}
A.~J. Kleywegt, A.~Shapiro, and T.~Homem-De-Mello.
\newblock The sample average approximation method for stochastic discrete optimization.
\newblock {\em SIAM Journal of Optimization}, 12:479--502, 2001.


\bibitem{MahdianYZ02}
M.~Mahdian, Y.~Ye, and J.~Zhang.
\newblock Approximation algorithms for metric facility location problems.
\newblock {\em {SIAM} Journal on Computing}, 36:411--432, 2006.

\bibitem{Markowitz52}
H.~M. Markowitz. 
\newblock Portfolio selection. 
\newblock {\em Journal of Finance}, 7:77--91, 1952.


\bibitem{NemirovskiS05}
A.~Nemirovski and A.~Shapiro.
\newblock Scenario approximations of chance constraints.
\newblock In G.~Calafiore and F.~Dabbene, editors. {\em Probabilistic and Randomized
  Methods for Design under Uncertainty}, Springer-Verlag, 2005. 

\bibitem{Prekopa73}
A.~Pr\'{e}kopa.
\newblock Contributions to the theory of stochastic programming.
\newblock {\em Mathematical Programming}, 4:202--221, 1973.

\bibitem{Prekopa95}
A.~Pr\'{e}kopa.
\newblock {\em Stochastic Programming}. 
\newblock Kluwer Academic Publishers, Dordrecht, 1995.

\bibitem{Prekopa03}
A.~Pr\'{e}kopa.
\newblock Probabilistic programming.
\newblock In A.~Ruszczynski and A.~Shapiro, editors, {\em Stochastic Programming}, volume
  10 of {\em Handbooks in Operations Research and Mgmt. Sc.}, North-Holland,
  Amsterdam, 2003. 

\bibitem{Pritsker97}
M.~Pritsker.
\newblock Evaluating value at risk methodologies. 
\newblock {\em Journal of Financial Services Research}, 12(2/3):201--242, 1997.
 
\bibitem{RaviS04}
R.~Ravi and A.~Sinha.
\newblock Hedging uncertainty: approximation algorithms for stochastic optimization problems.
\newblock {\em Mathematical Programming, Series A}, 108:97--114, 2006.


\bibitem{RockafellarU02}
R.~Rockafellar and S.~Uryasev. 
\newblock Conditional value-at-risk for general loss distributions. 
\newblock {\em Journal of Banking and Finance}, 26:1443--1471, 2002.

\bibitem{RuszczynskiS03}
A.~Ruszczynski and A.~Shapiro. 
\newblock Editors, {\em Stochastic Programming}, volume 10 of {\em Handbooks in Operations
  Research and Mgmt. Sc.}, North-Holland, Amsterdam, 2003. 


\bibitem{RuszczynskiS05}
A.~Ruszczynski and A.~Shapiro. 
\newblock Optimization of risk measures. 
\newblock In G.~Calafiore and F.~Dabbene, editors. {\em Probabilistic and Randomized
  Methods for Design under Uncertainty}, Springer-Verlag, 2005. 

\bibitem{Shapiro03}
A.~Shapiro.
\newblock Monte Carlo sampling methods.
\newblock In A.~Ruszczynski and A.~Shapiro, editors, {\em Stochastic Programming}, volume
  10 of {\em Handbooks in Operations Research and Mgmt. Sc.}, North-Holland,
  Amsterdam, 2003. 


\bibitem{ShmoysS06}
D.~B. Shmoys and C.~Swamy.
\newblock An approximation scheme for stochastic linear programming and its application to
  stochastic integer programs.
\newblock {\em Journal of the ACM}, 53(6):978--1012, 2006.

\bibitem{ShmoysS04}
D.~B. Shmoys and C.~Swamy.
\newblock Stochastic optimization is (almost) as easy as deterministic optimization. 
\newblock {\em Proceedings, 45th Annual {IEEE} {FOCS}}, pages 228--237, 2004. 

\bibitem{ShmoysTA97}
D.~B. Shmoys, {\'E}.~Tardos, and K.~I. Aardal.
\newblock Approximation algorithms for facility location problems.
\newblock {\em Proceedings, 29th Annual {ACM} Symposium on Theory of Computing},
  pages 265--274, 1997. 

\bibitem{SoZY06}
A.~M-C. So, J.~Zhang, and Y.~Ye.
\newblock Stochastic combinatorial optimization with controllable risk aversion level. 
\newblock {\em Proceedings, 9th APPROX}, pages 224--235, 2006.

\bibitem{Srinivasan07}
A.~Srinivasan.
\newblock Approximation algorithms for stochastic and risk-averse optimization.
\newblock {\em Proceedings, 18th SODA}, pages 1305--1313, 2007.

\bibitem{Swamy04}
C.~Swamy.
\newblock {\em Approximation Algorithms for Clustering Problems}.
\newblock Ph.D. thesis, Cornell University, Ithaca, NY, 2004.
  {\sffamily\small http://www.math.uwaterloo.ca/{\small $\sim$}cswamy/theses/master.pdf}.

\bibitem{SwamyS06}
C.~Swamy and D.~B. Shmoys.
\newblock Approximation algorithms for 2-stage stochastic optimization problems. 
\newblock {\em ACM SIGACT News}, 37(1):33--46, March 2006.
  Also appeared in {\em Proceedings, 26th FSTTCS}, pages 5--19, 2006. 


\bibitem{SwamyS05}
C.~Swamy and D.~B. Shmoys.
\newblock Sampling-based approximation algorithms for multi-stage stochastic
  optimization. 
\newblock {\sffamily\small http://www.math.uwaterloo.ca/{\small $\sim$}cswamy/papers/multistage-journ.pdf}. 
\newblock Preliminary version in {\em Proceedings, 46th Annual {IEEE} Symposium on
  Foundations of Computer Science}, pages 357--366, 2005.



\end{thebibliography}
\end{document}